\documentclass[12pt]{article}

\usepackage[english]{babel}
\usepackage{amstext}
\usepackage{amsmath,amssymb,amsthm,bm, float}
\usepackage{graphicx,psfrag,epsf}
\usepackage{enumerate}
\usepackage{color, soul}
\usepackage[colorinlistoftodos]{todonotes}
\usepackage{natbib}
\usepackage{url}
\usepackage{comment}

\newtheorem{proposition}{Proposition}
\newtheorem{theorem}{Theorem}
\newtheorem{remark}{Remark}
\newcommand{\blind}{1}
\newcommand{\V}{{\rm Var}}
\usepackage{enumitem}
\usepackage{ulem}
\usepackage{multirow}

\newcommand{\E}{{\rm E}}
\newcommand{\bD}{\boldsymbol{D}}
\newcommand{\upperRomannumeral}[1]{\uppercase\expandafter{\romannumeral#1}}

\numberwithin{equation}{section}

\let\OLDthebibliography\thebibliography
\renewcommand\thebibliography[1]{
  \OLDthebibliography{#1}
  \setlength{\parskip}{0pt}
  \setlength{\itemsep}{7pt}
}
\begin{document}

\def\spacingset#1{\renewcommand{\baselinestretch}%
{#1}\small\normalsize} \spacingset{1.00}

\if1\blind
{
\title{Calibration Concordance for Astronomical Instruments\\ via Multiplicative Shrinkage}
\author{Yang Chen\thanks{Yang Chen is Assistant Professor, Department of Statistics and Michigan Institute for Data Science (MIDAS), University of Michigan, Ann Arbor, MI 48109; email: ychenang@umich.edu.}, Xiao-Li Meng\thanks{Xiao-Li Meng is Whipple V. N. Jones Professor of Statistics, Harvard University, Cambridge, MA 02138.}, Xufei Wang\thanks{Xufei Wang was a Ph.D. candidate, Department of Statistics, Harvard University, Cambridge, MA 02138.}, David A. van Dyk\thanks{David A. van Dyk is a Professor of Statistics and Head of the Department of Mathematics at Imperial College London, London, UK SW7 2AZ.},\\ Herman L. Marshall\thanks{Herman Marshall is Astrophysicist, MIT Kavli Institute, Cambridge, MA 02139.}, Vinay L.\ Kashyap\thanks{Vinay Kashyap is Astrophysicist, Harvard-Smithsonian Center for Astrophysics, Cambridge, MA 02138.}}
  \maketitle
} \fi

\if0\blind
{
  \bigskip
  \bigskip
  \bigskip
  \begin{center}
{\LARGE\bf Calibration Concordance for Astronomical Instruments\\ via Multiplicative Shrinkage}
\end{center}
  \medskip
} \fi

\begin{abstract}
Calibration data are often obtained by observing several well-understood objects simultaneously with multiple instruments, such as satellites for measuring astronomical sources. Analyzing such data and obtaining proper concordance among the instruments is challenging when the physical source models are not well understood, when there are uncertainties in ``known'' physical quantities, or when data quality varies in ways that cannot be fully quantified. Furthermore, the number of model parameters increases with both the number of instruments and the number of sources. Thus, concordance of the instruments requires careful modeling of the mean signals, the intrinsic source differences, and measurement errors. In this paper, we propose a log-Normal model and a more general log-$t$ model that respect the multiplicative nature of the mean signals via a half-variance adjustment, yet permit imperfections in the mean modeling to be absorbed by residual variances. We present analytical solutions in the form of power shrinkage in special cases and develop reliable Markov chain Monte Carlo (MCMC) algorithms for general cases, both of which are available in the \texttt{Python} module \textit{CalConcordance}. We apply our method to several datasets including  a combination of observations of \textit{active galactic nuclei} (AGN) and spectral line emission from the \textit{supernova remnant} E0102, obtained with a variety of X-ray telescopes such as {\sl Chandra}, XMM-{\sl Newton}, {\sl Suzaku}, and {\sl Swift}. The data are compiled by the \textit{International Astronomical Consortium for High Energy Calibration} (IACHEC). We demonstrate that our method provides helpful and practical guidance for astrophysicists when adjusting for disagreements among instruments.  

\end{abstract}

\noindent%

\spacingset{1.45} 

\section{Introducing Calibration Concordance}
\label{section:introduction}

The calibration of instruments is fundamental for comparing or combining measurements obtained with different instruments. Typically, calibration is conducted by using each of several instruments to measure one or more well-understood objects, e.g., astronomical sources. The resulting data are used to develop  adjustments that can be applied to future observations to obtain reliable absolute measurements. Convenient adjustments, such
as ad hoc affine or ratio adjustments, however, often result in poor calibration, and without justifiable quantification of the calibration error that is essential for assessing the uncertainty of the final estimates of interest. The main difficulty of deriving reliable adjustments for instruments springs from the variations that are intrinsic to the sources and to the instruments, in addition to individual measurement errors. 

First, the physical models, derived using various approximations based on scientists' current understandings of the instruments, may not be as reliable as we hope. Second, ``known'' physical quantities are typically estimates themselves; even when their estimated errors are available, standard plug-in estimators and error propagation techniques may lead to biased and often overly optimistic results. Third, data quality varies in ways that cannot be fully quantified, especially across instruments or in the presence of outliers. Last, the number of unknown model parameters increases with the number of instruments and the number of sources, leading to well-known model challenges. Together these challenges and subtleties expose that, although calibration problems have a long history, principled statistical adjustments are not in routine use or even understood. 

This paper attempts to fill this gap for a variety of astronomical instruments, by developing  hierarchical models that respect the physical models for the mean signals, while permitting the modeling imperfections to be captured by residual variances. We build effective fitting algorithms 
and a software package, \textit{CalConcordance}, which are used to test our models via simulated data, and then applied to several datasets from the \textit{International Astronomical Consortium for High Energy Calibration} (IACHEC). The intended readers are both statisticians and astrophysicists.

\subsection{Calibration Concordance for Astronomical Instruments} 
In astrophysics, various instruments such as telescopes are used by different teams of scientists to understand intrinsic properties of astronomical objects, i.e., sources such as stars. Although it is possible to make relative comparisons of different sources observed with the same instrument, unless the instruments are properly calibrated~\citep{sembay2010defining}, we cannot make reliable absolute measurements or make comparisons of sources observed with different instruments. Therefore, calibration of different instruments is an important, and on-going, problem for astrophysicists~\citep[e.g.,][]{SewardLagacy1992,MatthewsProcSPIE2010, nevalainen2010cross,tsujimoto2011cross,read2014cross,schellenberger2015xmm,madsen2016iachec}. 

As an example, space-based (e.g., X-ray) telescope calibration~\citep{schwartz2014invited} is handled in two phases: first, under controlled laboratory conditions (``ground'' calibration), and second, while in space using astrophysical sources (``in-flight'' calibration; see~\citet{GuainazziJATIS2015}). At each phase, the same set of well-understood sources is observed with multiple instruments. The intent of in-flight calibration is usually to verify ground calibration but imperfect laboratory conditions and evolving instrument characteristics (while in-flight) may result in  discrepancies between different telescopes. The task of developing reliable adjustments for astronomical instruments based on observing multiple sources with multiple instruments is known as the \textit{calibration concordance problem}, which aims to develop a \textit{concordance} in the calibration among these astronomical instruments.  
This paper aims to provide a statistically principled solution and it is a joint effort between statisticians and astrophysicists, both of whose expertise are critical for  appropriately quantifying the uncertainties while incorporating scientific knowledge and judgments. 

For the calibration problem discussed in this paper, the following two concepts are essential.  
\begin{itemize}
\item \textbf{Flux of an astronomical source.} The absolute flux is the quantity of luminous energy incident upon the aperture of a telescope per unit area per unit time. The absolute flux of an astronomical source depends on the luminosity of the object and its distance from the Earth, both of which are intrinsic to the object. For a fixed source spectrum, i.e., the distribution of photon energies, the \textit{measured flux} is proportional to the number of photons \textit{detected} by an astronomical instrument. If the spectrum changes, or the detector on the instrument changes, then so do the number of detected photons and the measured flux.
\item \textbf{Effective Area for an instrument.} The geometric area of a telescope (instrument) is an upper bound on its capacity to collect photons. 
Many factors can reduce the efficiency of photon collection, including mirror reflectivity, structural obscuration, filter transmission, detector sensitivity, etc.  This reduction in efficiency is also photon-energy dependent. The \textit{Effective Area} is the equivalent geometric size of an ideal detector that would have the same collection capability and it is empirically measured or theoretically calculated and tabulated as a function of energy. An instrument's Effective Area is used to estimate the absolute flux of an astronomical source given its measured flux: the estimated absolute flux is the measured flux divided by the Effective Area. Since the Effective Area varies with photon-energy, astronomers often compare different energy bands in the way that we describe comparing different instruments, a convention we also adopt~\citep{GeorgeLagacy1992,GraessleProcSPIE2006}. 
\end{itemize}

The calibration problem arises because the Effective Areas of the instruments are not known precisely~\citep{drakea2006monte,kashyap2008handle,lee2011accounting,xu2014fully}, and hence different instruments can yield substantially different \textit{estimates of absolute fluxes} for the same unvarying source even after accounting for the measurement uncertainties in measured fluxes. This is manifested in Figure~\ref{fig:mediumraw}, which shows the logarithm of estimates of absolute fluxes of three sources (panels 1-3) using three instruments, ``pn'', ``MOS1'' and ``MOS2'', taken from the XCAL data that we describe in detail in Section~\ref{section:hmsdata}.
\begin{figure}[tbph]
\centering
\includegraphics[width=\textwidth]{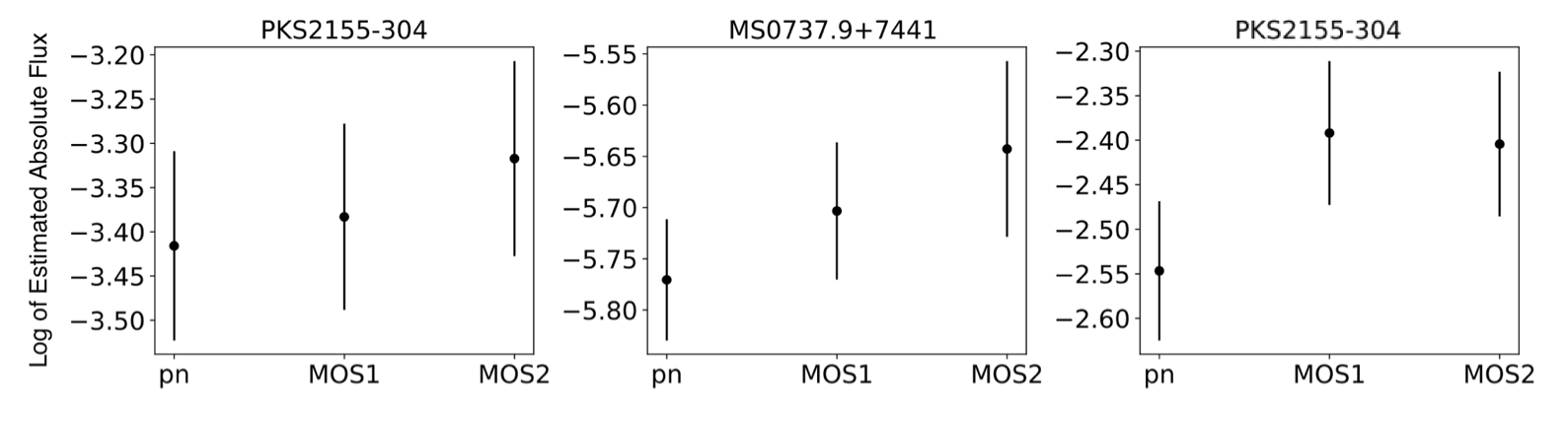}
\caption{Using (natural) logarithm of measured fluxes, correcting for existing ``known'' Effective Areas, to estimate log absolute fluxes. Measured fluxes are collected with three instruments ``pn, MOS1, MOS2'' for each of three sources, labeled on top of each panel; see Section~\ref{section:hmsdata}.  Estimates are given by the dots with approximate 95\% confidence intervals. The differences among the estimates are particularly pronounced in the third panel. The estimates from ``pn'' is systematically smaller than those from ``MOS1'' and ``MOS2'', illustrating the need for adjustments.}
\label{fig:mediumraw}
\end{figure}
Therefore, the problem of calibration concordance among different instruments is equivalent to reliably estimating the Effective Area of each instrument. By \textit{reliably} we mean that, after proper adjustments of the Effective Areas, instruments measuring a common source should agree within \textit{stated and scientifically acceptable} statistical uncertainty on the absolute flux of each source.

\subsection{A Multiplicative Physical Model}

Suppose we observe photon counts, $\{c_{ij}\}$, where $i$ indexes $N$ instruments and $j$ indexes $M$ objects/sources. The observed photon count $c_{ij}$ is known to follow a Poisson model with intensity $C_{ij}$, which is affected by the Effective Area $A_i$ and flux $F_j$ as follows. Because source fluxes have units of expected photons per second and per square centimeter, they are multiplied by instrument Effective Areas and a known factor, denoted by $T_{ij}$, to obtain expected photon counts:
\begin{equation}
\label{eqn:multiplicative}
C_{ij} = T_{ij} A_i F_j, \quad 1\leq i\leq N,\quad 1\leq j\leq M.
\end{equation}
The multiplicative constant $T_{ij}$ contains the exposure time, as well as other factors that can be calculated approximately by astrophysicists such as corrections for enclosed energy fractions and spectral shape correction factors; see~\citet{herman2017} for details.  With this in mind, we can regard $T_{ij}$ as a fixed known constant, and any real uncertainties related to $T_{ij}$ can be partially captured by our residual modeling discussed later in this paper and in subsequent work.  Intuitively, the $A_i$ can be regarded as a measure of the efficiency of instrument $i$ in terms of photon collection. Fundamentally, (\ref{eqn:multiplicative}) presumes that the Effective Area for a particular instrument remains the same regardless of which source it is applied to (and vice versa). By using a more homogeneous subgroup of sources (see Section~\ref{subsec:sampleselection}), we can increase the applicability of (\ref{eqn:multiplicative}), as we illustrate in Section~\ref{section:e0102data}. Of course, there is no free lunch -- by using a subgroup instead of all $M$ of the sources, we have less data and hence higher variability of our estimator, a bias-variance trade-off. 

Prior to observing $\{c_{ij}\}$, astronomers obtain initial estimates $a_i$ of $A_i$ from ground-based or in-flight calibration measurements, and hence it is safe to assume these measurements are independent of $\{c_{ij}\}$.
Comparing with estimated
fluxes of well-understood sources, astronomers can also place a reasonable prior bound on the margin of relative error in $a_i$   at about $20\%$~\citep{lee2011accounting,drakea2006monte}. Additional prior knowledge on the measurement errors in $\{c_{ij}\}$ is available. How to utilize this prior information, and whether these estimated uncertainties suffice to explain the variations in the data,  are among the questions that we investigate in this paper.

\subsection{Sample Selection Mechanism}
\label{subsec:sampleselection}

The sample selection mechanism, which involves both the selection of sources and the instrument used to observe each source is important because a biased selection mechanism can lead to misleading results. However, this is not a large concern in our setting for several reasons. 

First, the instruments we consider share a broad common energy passband so an object observed with one instrument will likely be seen with another. Indeed, for all the datasets we analyze in this paper, each source is observed by each instrument, although this is not a requirement for our methods. Second, the chance of a source not being observed because it is too faint is low. Since dim sources are not used for calibration, we do not include them in the study. Instead, we include sources with well-understood energy spectra, high intrinsic intensity, and stable spectral-temporal variations. This selection of sources is favorable since our ultimate goal is to calibrate the instruments and each of their Effective Areas is invariant to source fluxes. Furthermore, the vagaries of scheduling introduces large variations in the completeness achievable in a fleet of spacecraft, but this selection bias is negligible. Each spacecraft has several independent intrinsic constraints that are related to the shape of its orbit~\citep{chandrachap3}; its Sun, Moon, and Earth avoidance angles; and even its thermal environment histories; making scheduling of simultaneous observations difficult~\citep{chandrachap2}.  In other words, the missingness in the observation matrix, if any, is due to factors that are irrelevant to the intrinsic property of the sources or instruments, i.e., the estimands. These considerations permit us to ignore the sample selection mechanism in the sense of \citet{rubin1976inference}. 

The remainder of this paper is organized into 4 sections. Section~\ref{section:calibrationmodel} describes a statistical model for calibration concordance, a log-Normal model, and extends it to a more general log-$t$ model to handle outliers. Using simulated and real data, Sections~\ref{section:simulateddataresults} and \ref{section:realdataresults} assess and verify the empirical performance of our methods. Section~\ref{sec:discussion} briefly discusses a likelihood approach and its connection to our Bayesian approach, and future work. All numerical results are reproducible using the \textit{Python} code and data available on {\tt GitHub} at \url{https://github.com/astrostat/Concordance}. 

\section{Building and Fitting the Proposed Concordance Models}
\label{section:calibrationmodel}

\subsection{Modeling Multiplicative Means}
\label{subsec:addtivemean}

To make distinctions between observed quantities (e.g., estimator) and unknown quantities (e.g., estimand) clear, we adopt the convention that the former is denoted by lowercase (Roman) letters and the latter by uppercase, whenever feasible. We express (\ref{eqn:multiplicative}) as
\begin{equation}\label{eq:log}
\log C_{ij} - \log T_{ij}= B_i + G_j, \text{ where } B_i = \log A_i \text{ and } G_j = \log F_j.
\end{equation}
While this is a trivial relationship among the \textit{estimands}, it does not hold for their corresponding \textit{estimators}. In fact, if we let $y_{ij}=\log c_{ij} - \log T_{ij}$ (and ignore the issue of $c_{ij}=0$ for the moment), $b_i=\log a_i$ and $g_j=\log f_j$, we cannot simultaneously expect that $y_{ij}=b_i+g_j+\epsilon_{ij}$ \textit{and} that $\epsilon_{ij}$ is independent of $\{b_i, g_j\}$ with mean zero. If both were true, it would imply (incorrectly) that the expectation of $y_{ij}$ is determined by $b_i$ and $g_j$, rather than by their respective estimands: $B_i$ and $G_j$. Table~\ref{table:notationsEstimandEstimate} gives a summary of the notation used in this section.

The quantity that we observe and aim to model is $y_{ij} = \log c_{ij} - \log T_{ij}$, assuming $c_{ij} > 0$. (The case when $c_{ij}=0$, which never occurs in our data, is discussed below.) We assume that the measurement error in $c_{ij}$ for $C_{ij}$ is multiplicative (i.e., in terms of a percentage), which results in additive errors on the log-scale. Thus we postulate the regression model
\begin{equation}
\label{eq:threey}
 y_{ij} =  -0.5 \ \sigma_{i}^2 + B_i+G_j +e_{ij}, \quad e_{ij} \stackrel{{\rm indep}}{\sim} \mathcal{N}(0, \sigma_i^2), 
\end{equation}
where $-0.5 \ \sigma_{i}^2$ is a half-variance correction for the multiplicative mean modeling in (\ref{eqn:multiplicative}). This correction ensures that $\E(c_{ij})=C_{ij}$ because if $\log x \sim N(\mu, v)$, then 
$E(x)=e^{0.5v+\mu}$. Consequently, $\E(c_{ij})= {T}_{ij}\E(e^{y_{ij}})={T}_{ij} e^{0.5\sigma_i^2 - 0.5\sigma^2_{i}}e^{B_i}e^{G_j}= C_{ij}.$
For convenience, when (and \textit{only when}) $\sigma_i^2$ is known, we treat $y'_{ij}=y_{ij}+0.5\sigma^2_{i} $ as data. Since $b_i = \log a_i$ is an initial estimate of $B_i$ that is available without access to the calibration data $\{c_{ij}\}$,  we formulate it as the prior mean for $B_i$ via $
{B}_i {\rm \stackrel{indep}{\sim}} \mathcal{N}({b}_i, \tau_i^2)$,
where $\tau_i$ is provided by astronomers as well.

\begin{table}[t]
\begin{tabular}{c|ccccc}
\hline
& Counts & Effective Area & Flux & Log-Data & Correction \\
\hline
Estimand (Parameter) & $C_{ij}$ & $A_i =\exp(B_i)$ & $F_j = \exp(G_j)$ &   & $-0.5\sigma_i^2$ \\ 
Estimate (Data) & $c_{ij}$ & $a_i = \exp(b_i)$ &   & $y_{ij}$ & $T_{ij}$ \\
Relationship & ${E}(c_{ij}) = C_{ij}$ & $B_i\sim\mathcal{N}(b_i,\tau_i^2)$ & $C_{ij} = A_i F_j$ & $y_{ij} = \log(\frac{c_{ij}}{T_{ij}})$ & \\
\hline
\end{tabular}
\label{table:notationsEstimandEstimate}
\caption{Summary of notations for log-Normal model. The index $i$ ranges from $1$ to $N$ and the index $j$ ranges from $1$ to $M$.}
\end{table}

Given the underlying Poisson nature of the photon counts, $c_{ij}$, the log-Normal model in \eqref{eq:threey} deserves some explanation. If the expected counts, $C_{ij}$, are reasonably large, the log-Normal model approximates the Poisson model well. This is what we expect in practice since calibration sources are typically relatively bright, as illustrated in our datasets in Section~\ref{section:realdataresults}. However, the primary reason we adopt this approximation is because 
the log-Normal model permits separate modeling considerations for the mean and variance of the (transformed) counts,  whereas the Poisson mean dictates the Poisson variance. This flexibility is especially important when the mean model~\eqref{eqn:multiplicative} is not perfectly specified, as we expect since the $T_{ij}$ are estimated or approximated in practice~\citep{herman2017}. The variance of the log-Normal model can (partially) capture imperfections in the mean model; see Section~\ref{section:simulateddataresults} for discussions in our numerical experiments. 

Because the  log-Normal model works with the log counts, it cannot directly accommodate zero counts; we observe zero counts in some of our simulations studies. Should an observed count of zero be observed we suggest it be replaced by a pseudo count of $c_{ij} = 0.5$. This is a standard strategy ~\citep[p. 42]{bilder2014analysis} known as a zero-modified Poisson;  the mean and variance of a zero-modified Poisson random variable approximate those of the corresponding Poisson well if the mean of the Poisson is reasonable large.\footnote{For a Poisson random variable $q$ with mean $\lambda$, this replacement leads to a zero-modified Poisson random variable $\tilde{q}$ with mean and variance
\begin{equation}\label{eq:modified}
\E(\tilde q)=\lambda + 0.5 e^{-\lambda}, \qquad 
\V(\tilde q)= \lambda(1-e^{-\lambda})+ 0.25e^{-\lambda}(1-e^{-\lambda}).
\end{equation}
With reasonably large $\lambda$, $\tilde q$ approximates $q$ extremely well. } We validate this strategy in our simulation studies.

In (\ref{eq:threey}), the variance for the measurement error is assumed to depend only on the instrument. This assumption works reasonably well in our applied examples. More generally, each $e_{ij}$ can have its own variance, $\sigma^2_{ij}$, but obviously some constraints are needed in order to ensure identifiability. Other possible constraints include forcing the variances to be source-dependent only or to be additive, i.e., $\sigma_{ij}^2=\omega_i^2+\lambda_j^2$.  We first consider a \textit{known variance model} because astronomers provide best guesses of $\sigma_{ij}^2$. However, as illustrated in subsequent sections, the \textit{unknown variance model} is more flexible, robust, and hence recommended in practice. This is because the inferred adjustment of Effective Areas could be either overly-optimistic or overly-conservative if the specified $\sigma_{ij}^2$ are inaccurate. Unfortunately, this is often the case in practice owing to an incomplete quantification of measurement uncertainties or incomplete understanding of data preprocessing.

\subsection{Log-Normal Hierarchical Model and Its Posterior Sampling}
\label{section:hierregressionmodel}

Embedding the log-Normal model (\ref{eq:threey}) into a Bayesian hierarchical model requires a prior for $G_j$. Because astronomers do not know enough about the physical processes to place an informative prior on the fluxes, they prefer to use a flat prior for the log-scale flux, i.e., $G_j$, on the grounds that astronomical source fluxes cover many orders of magnitude in dynamic range ~\citep{appenzeller2012introduction}. 
When the $\sigma_i^2$ are treated as unknown, we adopt independent Inverse-Gamma distributions with shape parameter $\alpha$ and scale parameter $\beta$, the values of which are chosen to reflect the astronomers' prior knowledge about the approximate scale of noise levels. Specifically, we assume 
\begin{eqnarray}
\label{eqn:lognormalmodel}
&y_{ij}& | \boldsymbol{\ B,\ G,\ \sigma^2} \  {\rm\stackrel{indep}{\sim}} \ \mathcal{N}\left(-0.5\sigma_i^2 +B_i+G_j,\ \sigma_i^2\right),\\ 
\notag
&\sigma_i^2&{\rm\stackrel{indep}{\sim}} \ \text{Inv-Gamma}(\alpha,\ \beta),\quad
{B_i}  {\rm \stackrel{indep}{\sim}} N(b_i,\ \tau_i^2),\quad \text{ and } \quad G_j  {\rm\stackrel{indep}{\sim}} \ {\rm flat\ prior},
\end{eqnarray}
where $\boldsymbol{B} = (B_1,\ldots , B_N)^{\top}$, $\boldsymbol{G} = (G_1,\ldots, G_M)^{\top}$, $\boldsymbol{\sigma}^2 = (\sigma_1^2,\ldots,\sigma_N^2)^{\top}$, and $\boldsymbol{\tau}^2 = (\tau_1^2,\ldots,\tau_N^2)^{\top}$, and ${\top}$ denotes the usual transpose.  Under (\ref{eqn:lognormalmodel}), we can show that (see Appendix~\ref{section:properposterior}) the posterior distribution is proper with the weakest condition possible: each source is observed by at least one instrument. This theoretical guarantee is especially important because the number of parameters, $2N+M$, varies with the number of observations, $NM$.  Furthermore, the MAP (maximum-a-posterior) estimator of each $\sigma_i^2$ is bounded away from zero by a constant which depends only on the hyperparameters and the total number of sources (see Section~\ref{subsec:power}). Last,  the use of proper conjugate priors for $\boldsymbol{\sigma}^2$ avoids the problem of an unbounded posterior distribution, which can occur when we use uniform prior distributions for $\boldsymbol{\sigma}^2$. We also remark that, because each $\sigma_i^2$ enters both the variance and the mean in (\ref{eqn:lognormalmodel}), the impact of the choice of prior on the posterior inference is nuanced, as we discuss in the context of astrophysical applications (see Section~\ref{section:e0102data}).

In general, we let $J_i$ be the set of indexes of the objects observed by detector $i$ and $I_j$ be the set of indexes of the instruments that observe object $j$, and hence they accommodate missing data. Under \eqref{eqn:lognormalmodel},  the posterior density of $\{\boldsymbol{B, G, \sigma^2} \}$, if it exists, is proportional to 
\begin{equation}\label{eqn:joint}
\left[\prod_{i=1}^N \sigma_{i}^{-|J_i|-2-2 \alpha}\right] \exp\left\{- \frac{1}{2}\sum_{i=1}^N\sum_{j\in J_i} \frac{(y_{ij}+0.5\sigma^2_i-B_i-G_j)^2}{\sigma^{2}_{i}} -\sum_{i=1}^N \left[\frac{(b_{i}-B_i)^2}
{2\tau^2_{i}} +\frac{\beta}{\sigma_i^2}\right]\right\}.
\end{equation}
This implies that the conditional distribution of the column vector $\boldsymbol{\theta}\equiv(\boldsymbol{B}^\top, \boldsymbol{G}^\top)^\top$ given $\boldsymbol{\sigma^2}$ is 
an $(N+M)$-dimensional Normal density. A simple way of deriving the mean $\boldsymbol{\mu(\sigma^2)}$ and covariance $\boldsymbol{\Sigma(\sigma^2)}$ of this conditional distribution is to use partial derivatives of $L(\boldsymbol{\theta, \sigma^2})$, the logarithm of the joint posterior density given in (\ref{eqn:joint}). Let 
\begin{equation}\label{eq:normal}
\boldsymbol{\gamma(\sigma^2)}= \frac{\partial L(\boldsymbol{\theta, \sigma^2})}{\partial \boldsymbol\theta}\Big\vert_{\boldsymbol\theta= \boldsymbol 0}\quad {\rm and}\quad \boldsymbol{\Omega(\sigma^2)}= - \frac{\partial^2 L(\boldsymbol{\theta, \sigma^2})}{\partial \boldsymbol\theta^2} \Big\vert_{\boldsymbol\theta=\boldsymbol 0}.
\end{equation}
Note the second ``$|_{\boldsymbol\theta= \boldsymbol 0}$" is cosmetic because for Normal model the Fisher information is free of $\boldsymbol\theta$. 

By the form of the Normal density, $\boldsymbol{\mu(\sigma^2)=\Omega^{-1}(\sigma^2)\gamma(\sigma^2)}$ and $\boldsymbol{\Sigma(\sigma^2)=\Omega^{-1}(\sigma^2)}$. Evaluating these derivatives yields $\boldsymbol\gamma=(\gamma_1, \ldots, \gamma_N, \gamma_{N+1}, \ldots, \gamma_{N+M})^\top$  and $\boldsymbol\Omega$ as functions of $\boldsymbol{\sigma^2} $:
\begin{equation}\label{eq:gamma}
\gamma_i= \frac{\sum_{j\in J_i}y_{ij}}{\sigma^{2}_{i}}+\frac{b_i}{\tau^{2}_i}+0.5|J_i|,\ i=1,\ldots N, \quad \quad \gamma_{j+N}= \sum_{i\in I_j}\frac{y_{ij}}{\sigma^{2}_{i}}+0.5|I_j|, \ j=1, \ldots, M,
\end{equation}
\begin{equation}
\boldsymbol{\Omega} (\boldsymbol{\sigma}^2)=\left(\begin{array}{cc}
\bD_{N} & \bD \boldsymbol{R} \\ \boldsymbol{R}^\top \bD & \bD_{M} \end{array}\right), \quad {\rm where} \quad
\begin{array}{cc}
\bD_{N} =&{\rm Diag}\{|J_i|\sigma^{-2}_i+\tau^{-2}_i,\ i=1, \ldots, N \}, \\
\bD_{M} =&{\rm Diag}\{\sum_{i\in I_j}\sigma^{-2}_{i},\ j=1, \ldots, M\}, \end{array}\label{eqn:expressionomegasigma}
\end{equation}
$\bD={\rm Diag}\{\sigma^{-2}_i, \ i=1, \ldots, N\}$, and $\boldsymbol{R}=\{r_{ij}\}$, with $\boldsymbol{R}$ the $N\times M$ data ``recoding matrix'', i.e., $r_{ij}=1$ if source $j$ is observed with instrument $i$, and $r_{ij}=0$ otherwise.  When all the instruments measure all the sources, as in all our applications,  $\boldsymbol{R}$ and $\bD\boldsymbol{R}$ are rank-one. In such cases, the inverse of $\boldsymbol{\Omega}(\boldsymbol{\sigma}^2)$ can be calculated analytically, as seen in Appendix~\ref{appendix:derivationofomegainverse}. 

When $\boldsymbol{\sigma^2}$ is unknown, its marginal posterior density can be obtained by evaluating the identity $P(\boldsymbol\sigma^2)=P(\boldsymbol{\theta, \sigma^2})/P(\boldsymbol{\theta|\sigma^2})$ at $\boldsymbol\theta=\boldsymbol{0}$, where the numerator is given in~\eqref{eqn:joint} and the denominator is given by the conditional Normal distribution with mean and variance given in~\eqref{eq:normal}. For ease of notation, we use $P(\cdot)$ as simplified notation for the posterior density $P(\cdot|\left\{y_{ij}\right\})$. In particular, noting that 
$P(\boldsymbol{\theta}=\boldsymbol{0}|\boldsymbol{\sigma}^2)\propto \sqrt{|\boldsymbol{\Omega(\sigma^2)}|}\ e^{-\boldsymbol{\mu^\top(\sigma^2)\Omega(\sigma^2)\mu(\sigma^2)}/2}$,  $P(\boldsymbol\sigma^2) $ is proportional to
\begin{equation}\label{eqn:margin}
\prod_{i=1}^N \sigma_{i}^{-|J_i|-2-2 \alpha}\frac{1}{\sqrt{|\boldsymbol{\Omega(\sigma^2)}|}}\exp\left\{\frac{1}{2}\boldsymbol{\mu^\top(\sigma^2)\Omega(\sigma^2)\mu(\sigma^2)}- \sum_{i=1}^N \left[\frac{\sum_{j\in J_i} y_{ij}^2+2\beta}{2\sigma^{2}_{i}} +\frac{|J_i|\sigma^2_i}{8}\right]\right\}.
\end{equation}
Because this is not a standard distribution, numerical methods are required. We can obtain a Monte Carlo sample from the joint posterior in one of several ways, including applying an MCMC algorithm to sample $\{\boldsymbol{B},\boldsymbol{G},\boldsymbol{\sigma}^2\}$ jointly, or sampling $\boldsymbol{\sigma}^2$ from (\ref{eqn:margin}) via rejection sampling and $\{\boldsymbol{B},\boldsymbol{G}\}$ from its conditional Normal distribution given $\boldsymbol\sigma^2$. Incidentally, a good rejection proposal density is a convenient independent-component inverse Gamma distribution found in the proof of the posterior propriety; see Appendix \ref{section:properposterior}. The latter strategy is very efficient, especially as it provides independent draws. However, it is less flexible when we extend the model (e.g., the log-$t$ extension of Section~\ref{sec:logtmodel}).  Consequently, we adopt the more flexible MCMC approach.

Since the dimension of the parameter space, $2N+M$, is typically large for calibration purposes and the parameters are highly correlated, we use a Hamiltonian Monte Carlo (HMC) algorithm~\citep{HMC}, which delivers a less correlated sample than do other MCMC techniques~\citep[e.g.][]{MH1, MH2, Geman:1984}. We implement HMC using the \texttt{STAN} package in \textit{Python}~\citep{NUTS,RSTAN,pystanpkg}, along with a blocked Gibbs sampler as an independent cross-check of \texttt{STAN}. In the blocked Gibbs sampling, we sample $\{\boldsymbol{B, G} \}$ jointly to improve mixing as opposed to one-at-a-time Gibbs sampling, thanks to the joint normality of $\{\boldsymbol{B, G} \}$ conditioning on $\boldsymbol{\sigma^2}$; see Appendix~\ref{appendix:bayescomputation} for details. 

As is well known, computational efficiency for posterior sampling often is affected by modeling defects, such as near non-identifiability \citep[e.g.,][]{meng2018role}. Model (\ref{eqn:lognormalmodel}) does not suffer from this as long as $\tau_i$ is not too large compared to the magnitude of $\sigma_i$; see Appendix~\ref{appendix:identifiability}.

\subsection{Building Intuition: Power Shrinkage and Variance Shrinkage}\label{subsec:power}
To communicate our model aims clearly, when $\boldsymbol{\sigma}^2$ is known,  we express the MAP estimators of $\boldsymbol\theta$ in terms of the usual linear shrinkage estimators~\citep{shrinkage1,shrinkage2} of $\boldsymbol B$ and $\boldsymbol G$. Intuitively, shrinkage estimators combine information from all the instruments and sources as well as experts' prior information through weighted averages, which serve the purpose of calibration concordance across instruments and sources well. Specifically, by setting the derivative of the log posterior in (\ref{eqn:joint}) with respect to $\boldsymbol{B}$ and $\boldsymbol{G}$ to be zero, we find that the MAP estimators, conditional on $(\boldsymbol{\tau^2},\boldsymbol{\sigma}^2)$ and denoted by $\widehat{B}_i=\widehat{B}_i(\boldsymbol{\tau^2},\boldsymbol{\sigma}^2)$ and $\widehat{G}_j=\widehat{G}_j(\boldsymbol{\tau^2},\boldsymbol{\sigma}^2)$, must satisfy
\begin{equation}
\label{eqn:shrinkBG}
\widehat{B}_i (\boldsymbol{\tau^2},\boldsymbol{\sigma}^2) = W_i (\bar{y}_{i\cdot}' - \bar{G}_i)+(1-W_i) b_i \quad \text{ and } \quad  \widehat{G}_j (\boldsymbol{\tau^2},\boldsymbol{\sigma}^2) = \bar{y}_{\cdot j}' - \bar{B}_j,
\end{equation}
where for notational simplicity we write $y'_{ij}=y_{ij}+0.5\sigma^2_i$, suppressing the dependence on $\sigma^2_i$; $\bar y_{i\cdot}'$ is the precision (i.e., the reciprocal  variance) weighted average of the $y_{ij}'$ over $j\in J_i$, and $\bar y_{\cdot j}'$ is the precision weighted average of the $y_{ij}'$ over $i\in I_j$, i.e., $\bar y_{i\cdot}'=\frac{\sum_{j\in J_i}y_{ij}'\sigma^{-2}_{i}}{\sum_{j\in J_i}\sigma^{-2}_{i}}$ and $\bar y_{\cdot j}'=\frac{\sum_{i\in I_j}y_{ij}'\sigma^{-2}_{i}}{\sum_{i\in I_j}\sigma^{-2}_{i}}.$

In~\eqref{eqn:shrinkBG}, $\bar G_{i}$ is the precision weighted average of the $\widehat G_{j}(\boldsymbol{\tau^2},\boldsymbol{\sigma}^2)$ over $j\in J_i$ and $\bar{B}_j$ is the precision weighted average of the $\widehat{B}_i(\boldsymbol{\tau^2},\boldsymbol{\sigma}^2)$ over $j\in J_i$, i.e., 
\begin{equation*}
\bar G_{i}=\frac{\sum_{j\in J_i}\widehat G_{j}(\boldsymbol{\tau^2},\boldsymbol{\sigma}^2)\sigma^{-2}_{i}}{\sum_{j\in J_i}\sigma^{-2}_{i}}\quad \text{and} \quad \bar B_{j}=\frac{\sum_{i\in I_j}\widehat B_{i}(\boldsymbol{\tau^2},\boldsymbol{\sigma}^2) \sigma^{-2}_{i}}{\sum_{i\in I_j}\sigma^{-2}_{i}}.
\end{equation*}
Note that the expressions for $\widehat{B}_i(\boldsymbol{\tau^2},\boldsymbol{\sigma}^2)$ and $\widehat{G}_j(\boldsymbol{\tau^2},\boldsymbol{\sigma}^2)$ involve $\bar{G}_i$ and $\bar{B}_j$, which are linear combinations of $\widehat{B}_i(\boldsymbol{\tau^2},\boldsymbol{\sigma}^2)$ and $\widehat{G}_j(\boldsymbol{\tau^2},\boldsymbol{\sigma}^2)$. Therefore, when the variance parameters $\boldsymbol\sigma^2$ and $\boldsymbol\tau^2$ are known, \eqref{eqn:shrinkBG} and the expressions for $\bar{G}_i$ and $\bar{B}_j$ form a system of linear equations, the solutions of which are the MAP estimators for the $B_i$ and $G_j$. Finally, the weights 
\begin{equation}\label{eq:weight}
W_i=\frac{|J_i|\sigma^{-2}_{i}}{\tau^{-2}_i+|J_i|\sigma^{-2}_{i}}
\end{equation}
serve as the shrinkage factor for estimating $B_i$. The form of~\eqref{eq:weight} is intuitive because it measures the relative precision provided by the likelihood with respect to the total posterior precision.  Hence $1- W_i$ is the proportion of information from the prior distribution. This metric permits us to make judicious choices of the prior variances, $\boldsymbol{\tau}^2$, when they are not given by experts, so that our results are not unduly prior-driven. See Sections~\ref{section:e0102data},~\ref{section:2xmmhmsdata} and~\ref{section:hmsdata} for in-context discussions.

The linear shrinkage corresponds to 
``power shrinkage" on the original scale. Consider the case where $\boldsymbol{\tau^2}$ and all $G_j = g_j$ are known. In this case, (\ref{eqn:shrinkBG}) yields $\widehat B_i (\boldsymbol{\tau^2},\boldsymbol{\sigma}^2) = W_i(\bar y_{i \cdot}'-\bar g_{i})+(1-W_i)b_i.$ Consequently, the Effective Area is estimated by $$\widehat A_i =\widehat A_i (\boldsymbol{\tau^2},\boldsymbol{\sigma}^2)=\exp[\widehat B_i(\boldsymbol{\tau^2},\boldsymbol{\sigma}^2)] = a_i^{1-W_{i}} \left[(\tilde c_{i \cdot}\tilde f_i^{-1})e^{\sigma_{i}^2/2}\right]^{W_i},$$ 
where $\tilde c_{i\cdot}$ and $\tilde f_i$ are the geometric means: $\tilde c_{i \cdot}=\left[\prod_{j\in J_i} c_{ij}\right]^{1/|J_i|} \hbox{ and } \tilde f_i=\left[\prod_{j\in J_i} f_{j}\right]^{1/|J_i|}$. This adjustment depends on the relative precision $1-W_i$ for the $b_i$. If $W_i=1$, that is, if $b_i$ is not trustworthy at all, we ignore $a_i$ and estimate $A_i$ by $\widehat A_i=
\left[\tilde c_{i \cdot}\tilde f_i^{-1}\right] e^{\sigma_i^{2}/2}$. Note that the bias correction $e^{\sigma^2_i/2}$ is needed because otherwise $\tilde c_{i \cdot}\tilde f_i^{-1}$ converges to $A_i e^{-\sigma^2_i/2}$ as $|J_i|\rightarrow \infty$.  In contrast, if $W_i=0$, i.e., $b_i$ possesses no error,  then we ignore any data and just use $\hat A_i=a_i$ to estimate $A_i$. 

Because $W_i$ grows with $|J_i| \sigma^{-2}_i$, for fixed $\sigma^2_i$, the more calibration data we have, the larger the adjustment we make. However, the precision is not determined by the \textit{data size} $|J_i|$ alone, but also by the \textit{quality} of the data, as reflected in $\sigma^2_i$.  Hence if both $|J_i|$ and $\sigma^2_i$ are large, $W_i$ may not be near $1$ because the indirect information $|J_i| \sigma^{-2}_i$ may not be large compared to $\tau^{-2}_i$.   

When $\boldsymbol{\sigma}^2$ is unknown, we use the conjugate prior distributions for $\boldsymbol{\sigma}^2$ as in Section~\ref{section:hierregressionmodel}. Taking the derivative of the log of (\ref{eqn:joint}) with respect to $\sigma^2_i$ reveals that the MAP estimators also satisfy 
\begin{equation}\label{eqn:ssigma}
\widehat\sigma^2_i= 2\left[\sqrt{1+S^2_{y,i}}-1\right], \quad S_{y,i}^2= \frac{1}{|J_i|+ \alpha}\left[\sum_{j\in J_i}(y_{ij}-\widehat B_i-\widehat G_j)^2 + \beta\right],
\end{equation}
where $\beta$ is the shape parameter for the inverse Gamma prior distribution for $\sigma_i^2$ as given in (\ref{eqn:lognormalmodel}). 
We then solve (\ref{eqn:shrinkBG}) and (\ref{eqn:ssigma}) to obtain the MAP estimators $\{\boldsymbol{\widehat B, \widehat G, \widehat \sigma^2}\}$. For finite $|J_i|$, because $S^2_{y,i}\ge \beta/(|J_i|+\alpha)\ge \beta/(M+\alpha)$,  all $\widehat \sigma_i^2$ are bounded below by $2\sqrt{1+\beta/(M+\alpha)} - 2>0$. {Hence our model, including its prior specifications, avoids the problem of an unbounded likelihood at $\sigma_i^2=0$, which is a known problem of hierarchical modeling with weak likelihood or prior information. }

Intriguingly, the MAP estimator for the variance is also a shrinkage estimator,
\begin{equation*}
\widehat\sigma^2_i=2\left[\sqrt{1+S^2_{y,i}}-1\right] = \frac{2}{1+\sqrt{1+S^2_{y,i}}}\ S_{y, i}^2 \equiv  R_i S^2_{y,i},
\end{equation*}
where $S_{y, i}^2$ of (\ref{eqn:ssigma}) is a natural extension of residual variance estimator, incorporating  prior information through $\{\alpha, \beta\}$. The half-variance correction leads to a shrinkage of $S^2_{y,i}$ because $R_{i}\le 1$. The degree of shrinkage depends on $S^2_{y,i}$ itself: the larger $S^2_{y, i}$ is, the more shrinkage. Such a self-weighted non-linear shrinkage phenomenon appears to be new.

\subsection{Extensions to Handling Outliers: Log-$t$ Model}
\label{sec:logtmodel}

Outliers are not uncommon in astronomical observations because 
the harsh environments in which the detectors operate can be subject to large variations in background intensities, potentially
leading to large errors in flux estimates.  In addition, astronomical sources have intrinsic variabilities covering many orders of magnitude, and some measurements could be performed in regimes where the detectors do not respond linearly to the incoming signal. For these reasons, we propose a robust log-$t$ model as a generalization of the log-Normal model to better handle outliers (see, e.g.~\citet{lange1989robust}). Specifically, we introduce a latent variable $\xi_{ij}$ for each observation $y_{ij}$ that is used to down-weight outliers. Formally, for each observation $y_{ij}$, we assume
\begin{eqnarray}
\label{eqn:logtmodel}
y_{ij}\mid \boldsymbol{\ B,\ G,\  \xi}  &= & -\frac{\kappa^2}{2\xi_{ij}} +B_i+G_j + \frac{Z_{ij}}{\sqrt{\xi_{ij}}},\\ 
\notag
Z_{ij}|\boldsymbol{\xi} & \rm{\stackrel{indep}{\sim}} & \mathcal{N}(0, \kappa^2),\quad \text{and} \quad  B_i {\rm\stackrel{indep}{\sim}}  \mathcal{N}(b_i, \tau_i^2),
\end{eqnarray}
where $\boldsymbol{\xi}=\{\xi_{ij}\}$. Because $\E(e^{y_{ij}}\mid \boldsymbol{B,  G}) = \E\left[ \E (e^{y_{ij}}\mid \boldsymbol{B,\ G,\  \xi}) | \boldsymbol{B, G} \right] = A_iF_j$, the multiplicative model in (\ref{eqn:multiplicative}) is maintained. Depending on the assumptions made for $\xi_{ij}$, (\ref{eqn:logtmodel}) includes: 
\begin{itemize}[label={}]
\item \textit{Case 1}: \textit{log-Normal model with known variances}. If the $\xi_{ij}$ are known constants, the noise terms $e_{ij} = Z_{ij} /\sqrt{\xi_{ij}}$ are independent Normals with mean $0$ and variance $\sigma_{ij}^2={\kappa^2}/{\xi_{ij}}$. Thus the model in (\ref{eqn:logtmodel}) is equivalent to \eqref{eqn:lognormalmodel} with known variances.
\item \textit{Case 2}: \textit{log-Normal model with unknown variances}. If $\xi_{ij}=\xi_i  \stackrel{\rm indep}{\sim}\chi_{u}^2$ for all $j$, then the variances of $Z_{ij} /\sqrt{\xi_i}$ conditional on $\xi_i$ are ${\kappa^2}/{\xi_i}$, which are distributed as independent scaled inverse $\chi^2$ with degree of freedom $u$ and scale $u\kappa^2$. Thus (\ref{eqn:logtmodel}) is equivalent to \eqref{eqn:lognormalmodel} with $\alpha = u/2$ and $\beta = {\kappa^2}/{2}$, noting the equivalence between $\chi^2$ and Gamma distributions.
\item \textit{Case 3}: \textit{log-$t$ model}. If $\xi_{ij} \rm{\stackrel{indep}{\sim}} \chi^2_{\nu}$, i.e., mutually independent $\chi^2$ random variables, which are also independent of the $Z_{ij}$, then the error terms $Z_{ij} / \sqrt{\xi_{ij}}$ follows independent (scaled) student-$t$ distributions: ${Z_{ij}}/{\sqrt{\xi_{ij}}}\ \stackrel{\rm indep}{\sim}\  \left({\kappa}/{\sqrt{\nu}}\right)\ t_{\nu}.$ 
\end{itemize}
All of these models can be fit using HMC via \texttt{STAN} in the package \textit{CalConcordance}.

Besides down-weighting outliers, the latent $\boldsymbol\xi$ also permits a unique variance $\kappa^2/\xi_{ij}$ for each instrument-source combination. The log-$t$ model is thus more flexible, but computationally more demanding, than the log-Normal model: convergence of HMC is harder to achieve and the sampling is more costly. When the $\xi_{ij}$ are small, the half variance corrections ${\kappa^2}/{2\xi_{ij}}$ are also likely to dominate the error terms ${Z_{ij}}/{\sqrt{\xi_{ij}}}$, because of their (much) smaller denominator. This results in  small $y_{ij}$,  and hence the model tends to generate heavier left tails than right tails.

Despite these challenges, which are topics for further study, we demonstrate the effectiveness of log-$t$ compared with log-Normal model for simulated and real data in the presence of serious outliers. Among heavy-tailed distributions, we choose log-$t$ because it is a natural extension of the log-Normal model and it carries the intuitive interpretation of ``down-weighting'' outliers: $\xi_{ij}$ serves as the ``weight'' for $y_{ij}$. The last point is confirmed in Simulation~\upperRomannumeral{3} in Section~\ref{subsubsec:simulationsfortmodel}: outliers have much smaller estimated $\xi_{ij}$ relative to other observations. Without outliers, however, we recommend the log-Normal model for its adequacy and computational simplicity.

\section{Testing the Concordance Models with Simulated Data}
\label{section:simulateddataresults}

Our simulation studies aim to demonstrate that 
(1) the log-Normal model is reasonably robust to the type of model specifications likely to occur in practice; (2) a commonly adopted plug-in method treating guesstimated variances as known can lead to very 
poor adjustments; and (3) the log-$t$ model is preferred in the presence of serious outliers. We choose the simulation sample sizes to be on the same order as those in our applied examples to make the results more interpretable.

\subsection{Checking Robustness to Likely Misspecification}
\label{section:robustnesspoisson}

As seen in Section~\ref{section:calibrationmodel}, we approximate the Poisson counts via a log-Normal distribution and treat the $T_{ij}$ as known quantities in (\ref{eqn:multiplicative}). These assumptions have reasonable justifications \citep{herman2017}, but nevertheless we should exercise due diligence. Here we study the adequacy of the log-Normal approximation. Further simulations (\upperRomannumeral{4}-\upperRomannumeral{6}) are presented in Appendix~\ref{appendix:robustness_simulation} to investigate the effect of treating $T_{ij}$ fixed and how it interacts with the log-Normal approximation.

In Simulations~\upperRomannumeral{1} and~\upperRomannumeral{2}, there are $N=10$ instruments and $M=40$ sources. The data are generated as $y_{ij}=\log \tilde c_{ij}$, where $\{\tilde c_{ij}\}$ are independent zero-modified Poisson counts with $\lambda_{ij}=A_iF_j=\exp(B_i+G_j)$. We set $\{B_i=1, G_j=1\}$ in Simulation~\upperRomannumeral{1} and $\{B_i=5, G_j=3\}$ in Simulation~\upperRomannumeral{2}. We independently sample $b_i = \log a_i$ from $\mathcal{N}(B_i, 0.05^2)$. Thus Simulation~\upperRomannumeral{1} represents a low count scenario where the log-Normal approximation may not be appropriate. When $\boldsymbol\sigma^2 = 0.1^2 \boldsymbol{1}_N$, where $\boldsymbol{1}_N$ denotes an $N\times 1$ column vector of $1$s,  we use the posterior for $B_i$ and $G_j$ given in Section~\ref{section:hierregressionmodel}. Otherwise, we specify the prior of each $\sigma_i^2$ as independent inverse Gammas with degree of freedom $\alpha=2$ and scale $\beta=0.01$,  and use HMC to obtain draws from the joint posterior distribution. Our choice of hyperparameters are set to match astronomers' prior knowledge; for example,  $0.1^2$ is the maximum of their guesstimates of $\boldsymbol{\sigma}^2$,  reflecting the general consensus  that $10\%$ or less relative  error (recall $y_{ij}$ is on log scale) does not alter physical interpretations in important ways. As expected, the fitted values of $B_i$ and $G_j$ are much closer to their targets in Simulation~\upperRomannumeral{2} since it has more Poisson counts resulting from larger values of $B_i$ and $G_j$; Figures~\ref{fig:simpoissonN10M10B1G1} and~\ref{fig:simpoissonN10M10B5G3} give detailed results (for $B_i,
\sigma_i$) under Simulations~\upperRomannumeral{1} and~\upperRomannumeral{2}. 

\begin{figure}[t]
\centering
\includegraphics[width=0.9\textwidth]{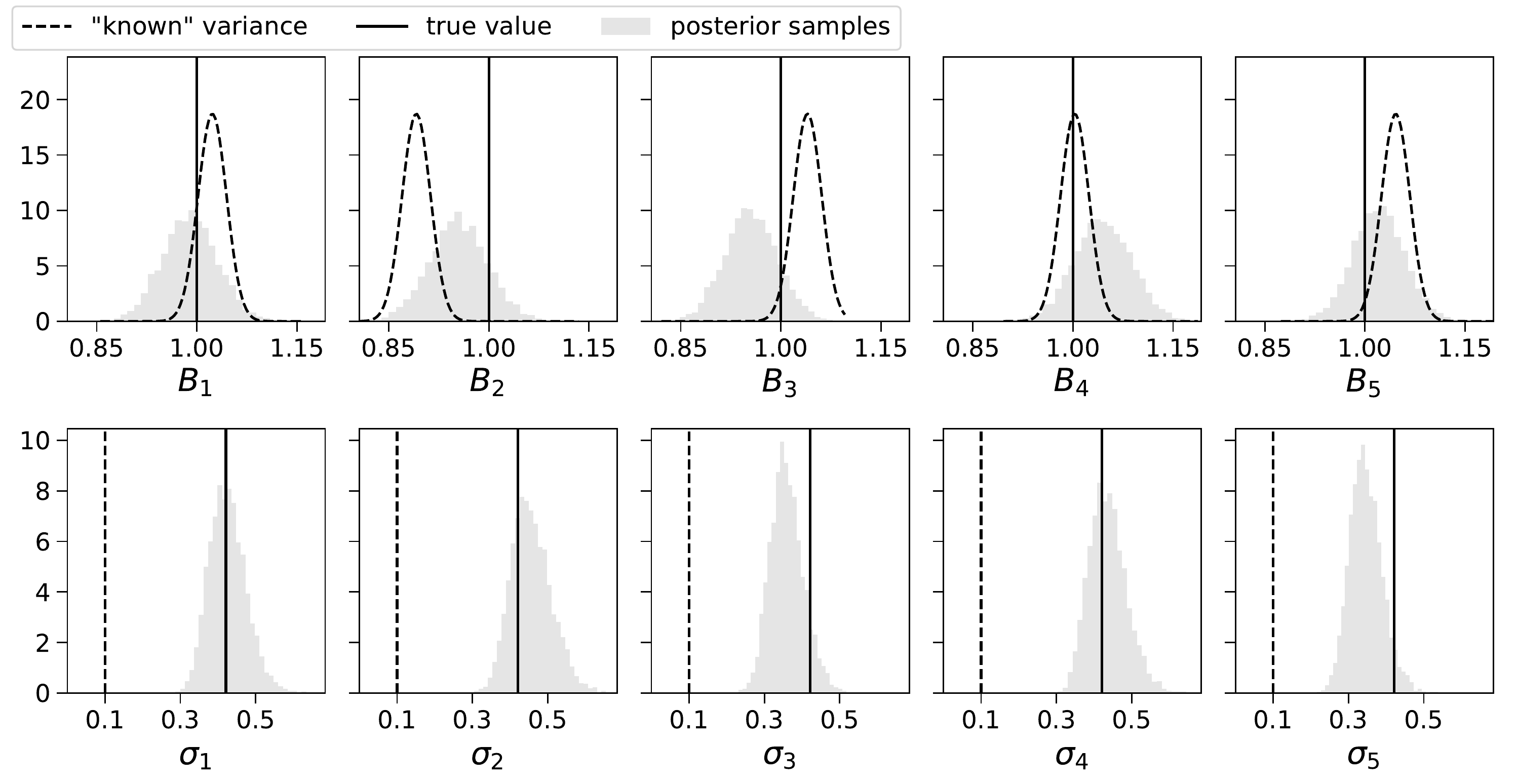}
\caption{Simulation~\upperRomannumeral{1}. Posterior histograms of $\{B_i,\sigma_i\}_{i=1}^5$. The solid vertical black lines denote the true/theoretical values of $B_i=1$ (top row)  and of $\sigma_i=0.421$ (second row). The dashed vertical lines denote $\sigma_i=0.1$ (second row). The black dashed density curves denote the exact posterior densities of $B_i$ when we set the variances equal to their guesstimated value $\sigma_i^2=0.1^2$.}
\label{fig:simpoissonN10M10B1G1}
\end{figure}

\begin{figure}[tbph]
\centering
\includegraphics[width=0.8\textwidth]{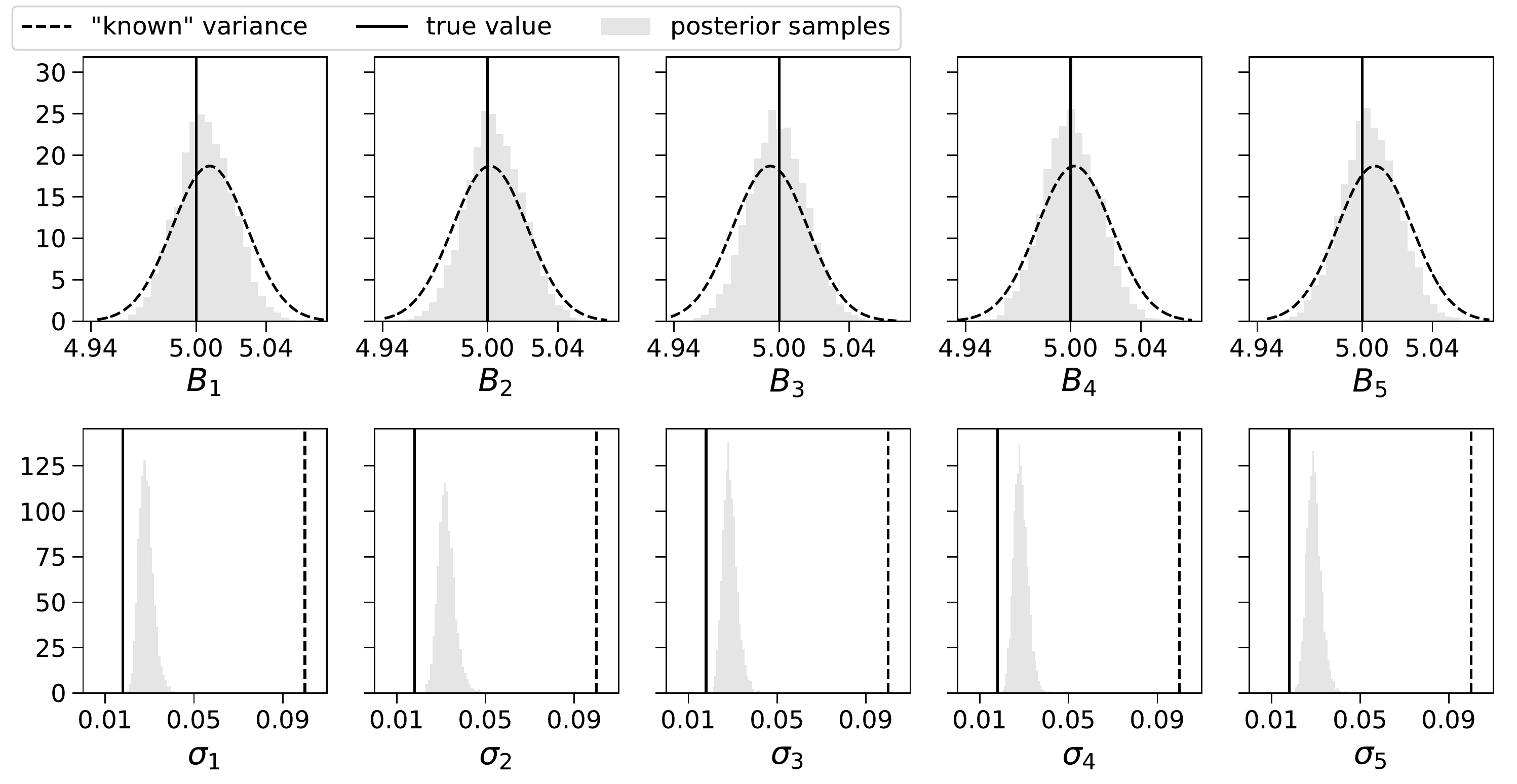}\\
\includegraphics[width=0.8\textwidth]{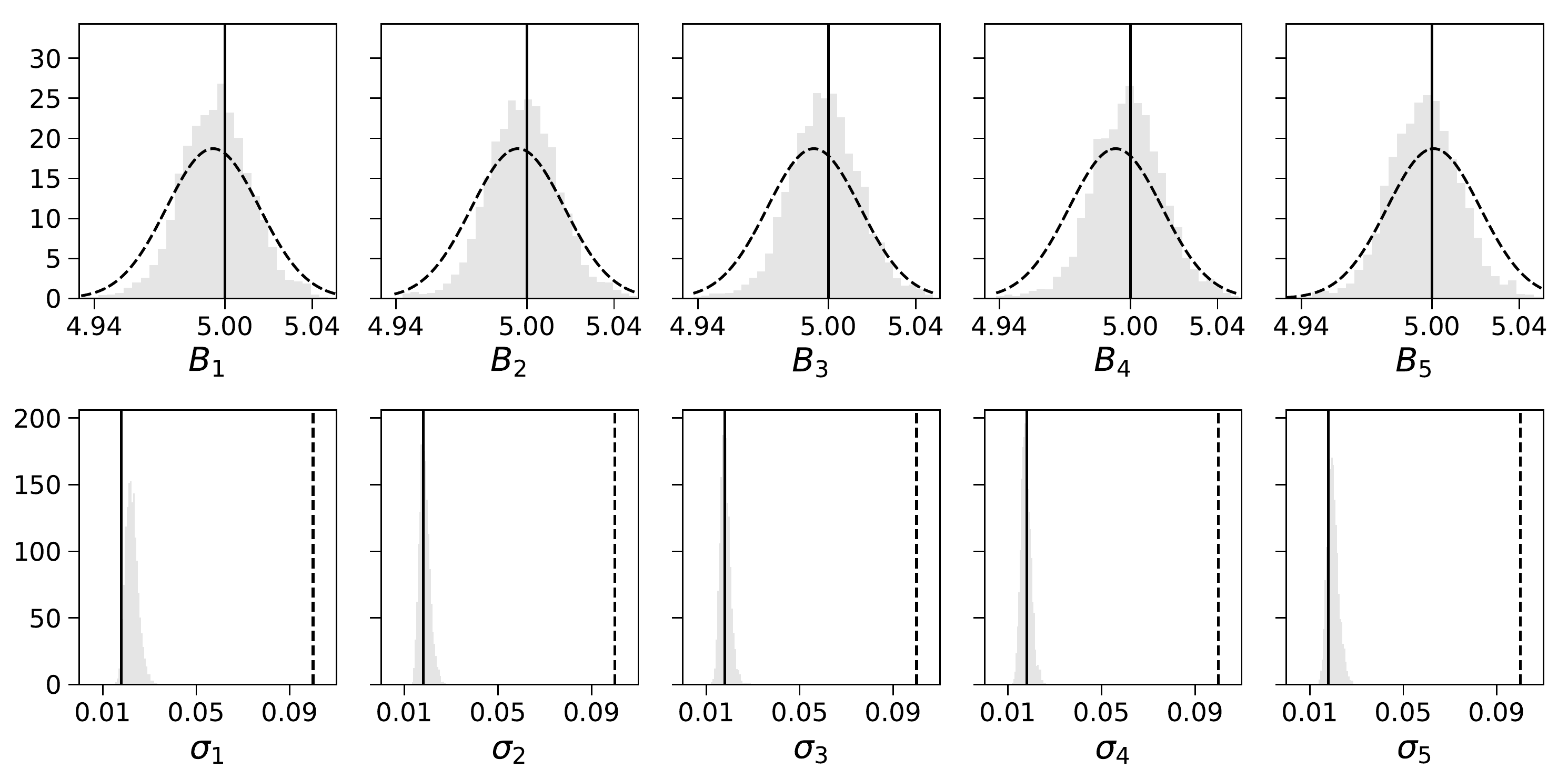}
\caption{Simulation~\upperRomannumeral{2}. Posterior histograms of $\{B_i,\sigma_i\}_{i=1}^5$ corresponding to $\beta = 0.01$ (rows 1 \& 2) and $\beta = 0.001$ (rows 3 \& 4). The solid vertical black lines denote the true values of $B_i=5$ (rows 1 \& 3) and of 
$\sigma_i=0.018$ (rows 2 \& 4). The dashed vertical lines denote the guesstimated value $\sigma_i=0.1$ (rows 2 \& 4). The black dashed density curves in rows 1 \& 3 denote the exact posterior densities of $B_i$ when we set the variances equal to their guesstimated value $\sigma_i^2=0.1^2$. }
\label{fig:simpoissonN10M10B5G3}
\end{figure}

Given $B_i$ and $G_j$, since $\tilde c_{ij}$ is discrete, the theoretical variance $\V[\log(\tilde c_{ij})]$, which is a function of $\lambda=\exp(B_i+G_j)$, can be calculated numerically to any desired accuracy. For reasonably large $B_i+G_j$, we can also approximate $\sigma_{ij}^2$ by the $\delta$-method: 
$\sigma^2_{ij} \approx \V(\tilde c_{ij})/\E^2(\tilde c_{ij})$, where $\E(\tilde c_{ij})$ and $\V(\tilde c_{ij})$ are obtained from (\ref{eq:modified}) with $\lambda=\exp(B_i+G_j)$.
For Simulations~\upperRomannumeral{1}, $\lambda=e^2=7.4$, which leads to $\sigma_{ij}=0.421$ numerically; in contrast, the $\delta$-method gives $\sigma_{ij}\approx 0.367$, a poor approximation due to the smallness of $\lambda$. This is a warning as to the inadequacy of using the log-Normal approximation. In contrast, for Simulation~\upperRomannumeral{2}, $ \lambda = e^8= 2981$, and hence $\sigma_{ij}=0.018$;
the $\delta$-method gives the same figure (to four significant digits). Note that the priors for $\boldsymbol{\sigma}^2$ are inverse Gammas with degrees of freedom $2$ and scale $\beta$ for both simulations; we use $\beta = 0.01$ for Simulation~\upperRomannumeral{1}, and $\beta = 0.01$ as well as $\beta=0.001$ for Simulation~\upperRomannumeral{2}.

Now suppose we set each $\sigma^2_i = 0.1^2$ as guesstimates of the variances. Comparing the histograms and the overlaying curves from Figures~\ref{fig:simpoissonN10M10B1G1} and~\ref{fig:simpoissonN10M10B5G3}, we see in Simulation~\upperRomannumeral{1} that the posterior distributions largely miss their targets, because $\sigma_i^2 = 0.1^2$ is significantly smaller than the variances estimated under the unknown variance model. In contrast, when the true value of $\boldsymbol{\sigma}^2$ is larger than its guesstimate, as in Simulation~\upperRomannumeral{2}, the posterior distributions of the $B_i$ do capture the target, but exhibit longer tails compared to those resulting from estimated variances.  Both phenomena are expected and confirm that if one must guesstimate the variances, it is better to err on the conservative side. Of course, larger variances imply less precision, which leads to less informative results, an inevitable but small price for overestimating the variances.

\subsection{Dealing with Outliers via log-$t$ Model}
\label{subsubsec:simulationsfortmodel}

Simulation~\upperRomannumeral{3}, which is the same as Simulation~\upperRomannumeral{2} except we set $G_1=-2$ to induce outliers, demonstrates the effectiveness of the log-$t$ model in dealing with outliers. Setting $G_1 = -2$ leads to more outliers in the first source as compared to other sources with $G_j = 3$ because the data generating model is Poisson with count rate $\exp(B_i + G_j)$. Under this model the variance of the logarithm of counts is approximately $e^{-B_i - G_j}$ when $B_i+G_j$ is large. Thus, a very small $G_1$ yields a much larger variance relative to the other sources (by a magnitude of $e^{3+2}\approx 150$) and more extreme observed values. This is a realistic mechanism for generating outliers because it represents the case where one of the sources is much fainter than the others. Following the notation in Case 3 of Section~\ref{sec:logtmodel}, setting the shape parameter for the inverse Gamma prior to $\alpha=2$ is the same as setting $\nu=2\alpha=4$ in a $\chi^2_\nu$ prior; the scale parameter $\beta=0.01$ corresponds to $\kappa=\sqrt{2\beta} \approx 0.141$. In this section, we use $\sigma_{ij}^2$ to denote the residual variance for $y_{ij}$. From Simulation~\upperRomannumeral{2}, $\sigma_{ij}=0.018$ for $j>1$. For $j=1$, $\lambda=e^{-2+5}=20.1$, and hence the exact numerical calculation gives $\sigma_{i1}=0.232$, and the $\delta$-method yields $\sigma_{i1}\approx0.223$, which is quite a reasonable approximation. 

Using the first three sources as an example, the upper panel in Figure~\ref{fig:logtlognormalcomparesim3} compares the results from the log-Normal model and the log-$t$ model through the (fitted) standardized residuals, which are given respectively by (using the notation in Section~\ref{section:hierregressionmodel} and Section~\ref{sec:logtmodel}):
\begin{equation}
\widehat{\mathcal{R}}_{ij}=\frac{y_{ij}-\widehat{B}_i-\widehat{G}_j + 0.5 \times \widehat{\sigma}_{i}^2}{\widehat{\sigma}_{i}} \quad {\rm and}
\quad \widehat{\mathcal{R}}_{ij} = \frac{y_{ij}-\widehat{B}_i-\widehat{G}_j + 0.5 \times \kappa^2/\widehat{\xi}_{ij}}{\kappa/\widehat{\xi}_{ij}^{1/2}},
\label{eqn:standardizedresidue}
\end{equation}
where $\widehat{B}_i, \widehat{G}_j, \widehat{\xi}_{ij}$ and $\widehat{\sigma}_i$ are the posterior means. We see some observations from the first source (black circles) are outliers with standardized residuals outside $[-2,2]$, in the log-Normal model (upper panel) but not for the log-$t$ model (lower panel). In the log-Normal model, setting $\sigma_{ij}^2=\sigma^2_i$ causes large standardized residuals due to some source-dependent large variances: $\sigma_{i1}^2 >> \sigma_{ij}^2$, $j\geq 2$. In the log-$t$ model, the outliers are down weighted by $\xi_{ij}$ and each observation is assigned a unique conditional variance $\sigma_{ij}^2=\kappa^2/\xi_{ij}$,  illustrating the benefit of
using the log-$t$ model. 

\begin{figure}[t]
\centering
\centering\includegraphics[width = 0.65\textwidth]{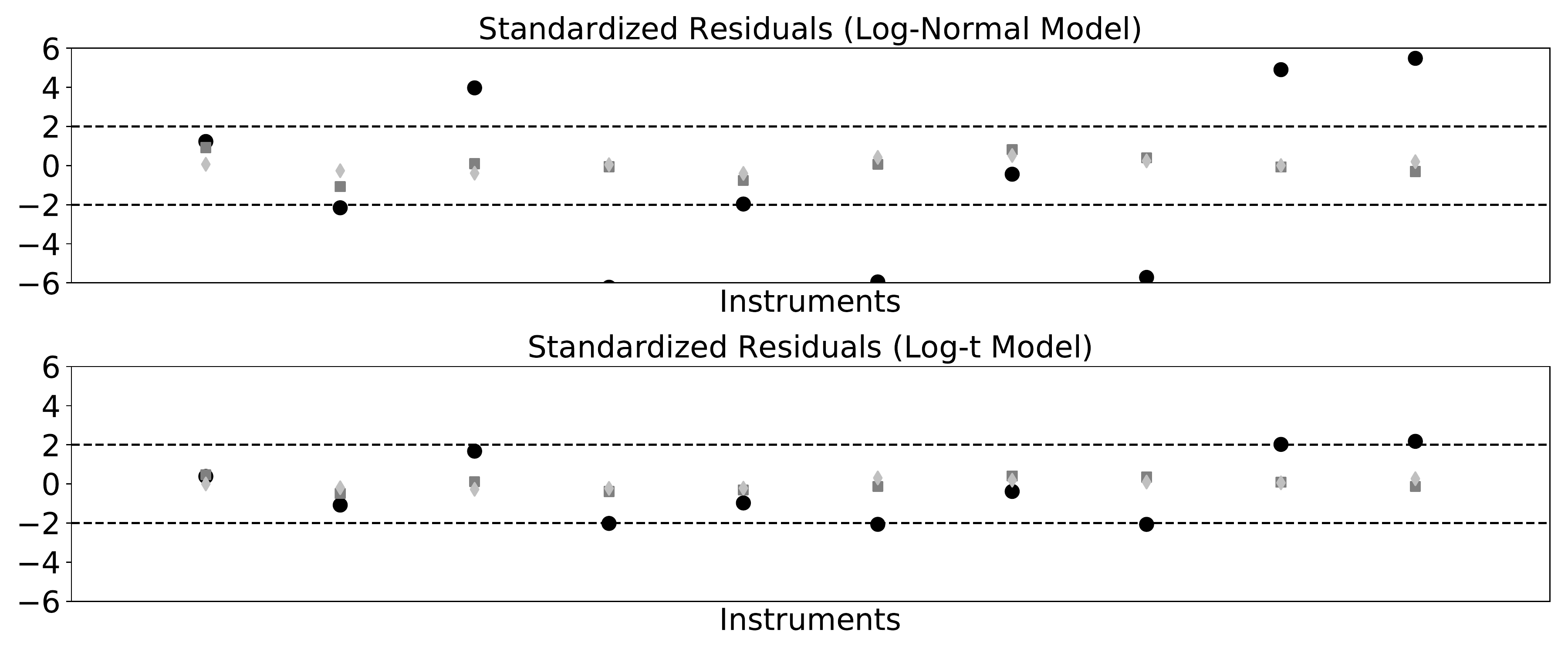}\\
\includegraphics[width = 0.95\textwidth]{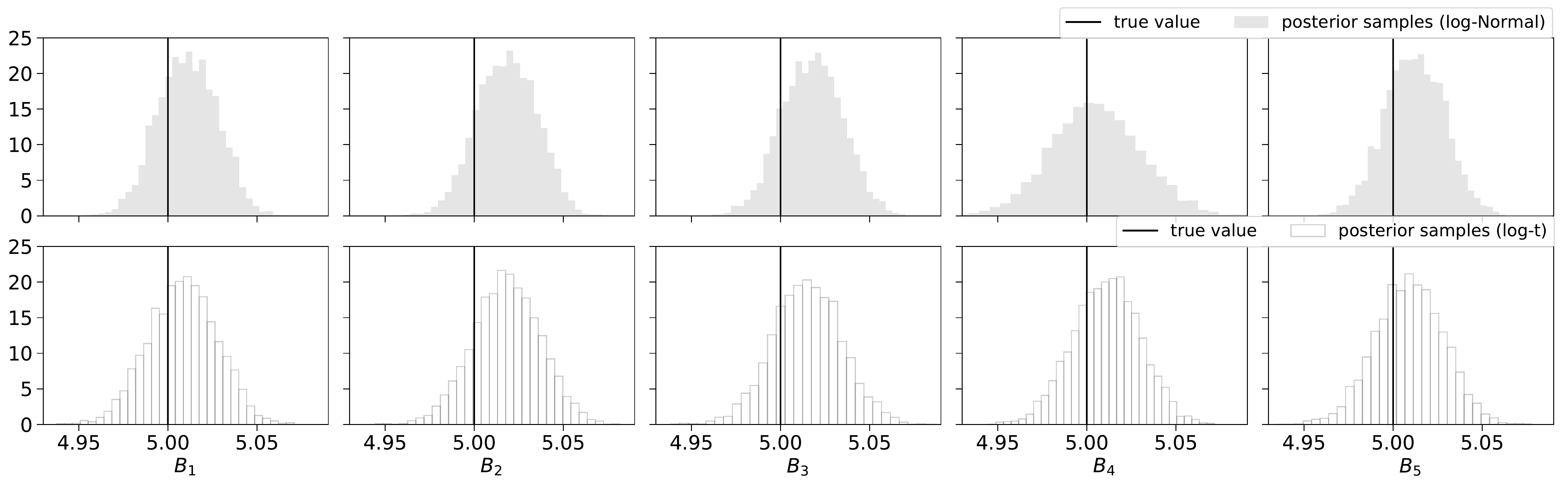}
\caption{Simulation~\upperRomannumeral{3}. Rows 1 and 2 show standardized residuals of the log-Normal model and the log-$t$ model. The black circles, gray squares and silver rhombi represent the first three sources respectively. The instruments are plotted on the x-axes. The dashed horizontal lines denote the $[-2,\, 2]$ intervals. Rows 3 and 4 show posterior histograms of $\{B_i\}_{i=1}^5$ from the log-Normal model and the log-$t$ model, where the black vertical bars indicate true values.} 
\label{fig:logtlognormalcomparesim3}
\end{figure}

The lower panel in Figure~\ref{fig:logtlognormalcomparesim3} shows the posterior distributions of $B_i$ under the log-Normal and log-$t$ models, both capturing the true value. The log-$t$ model exhibits slightly larger variances between the two. This is comforting, especially considering the flexibility of the log-t model, permitting individual $\sigma_{ij}$ rather than the hard-constraint $\sigma_{ij}=\sigma_i$ of the log-Normal model.

The first three rows of Table~\ref{table:coveragecompare} give the (average) coverage of nominal equal-tailed $95\%$ posterior intervals for $B_i$ and $G_j$ obtained from $2000$ simulations with the same configurations as in Simulation~\upperRomannumeral{3}. 
The log-$t$ model is more robust to outliers than the log-Normal model, exhibiting significantly better coverage for $G_1$. (Though coverage under the log-$t$ model is still poor relative to the nominal level.) Table~\ref{table:coveragecompare} also indicates that outliers are not as problematic for estimating $B_i$, our primary interest, as they are for $G_1$. This is because we are more ``informed" about the $B_i$ than the $G_j$ since (1) in this experiment $N=10, M=40$, and hence there are more sources than instruments, and (2) each $B_i$ has an informative prior whereas each $G_j$ only has a flat prior.

\begin{table}[t]
\centering
\begin{tabular}{c|c|c|c|c|c}
\hline
Data Generating & \multirow{2}{*}{Parameter} & \multicolumn{2}{c|}{Coverage Probability} & \multicolumn{2}{c}{Length of Interval}\\ \cline{3-6}
Model (Poisson) && log-Normal & log-$t$ & log-Normal & log-$t$\\ \hline
$N=10,M=40$ & $\boldsymbol{B}$ & [0.941, 0.959] & [0.971, 0.975] & 0.067$\pm$0.005 & 0.073 $\pm$ 0.002\\ \hline
$N=10,M=40$ & $G_1$ & \textit{0.399} & \textit{0.700} & \textit{0.090$\pm$ 0.015} & \textit{0.182$\pm$0.045}\\ \hline
$N=10,M=40$ & $G_2, \ldots, G_M$ & [0.967, 0.977] & [0.996, 0.999] & 0.077$\pm$0.003 & 0.104$\pm$0.002\\ \hline
$N=40,M=40$ & $\boldsymbol{B}$ & [0.953, 0.969] & [0.993, 0.998] & 0.041$\pm$0.007 &0.050$\pm$0.001 \\ \hline
$N=40,M=40$ & $G_1$ & \textit{0.398} & \textit{0.686} & \textit{0.045$\pm$0.003} & \textit{0.093$\pm$0.013} \\ \hline
$N=40,M=40$ & $G_2, \ldots, G_M$ & [0.965,0.977] & [0.996,0.999] & 0.038$\pm$0.001 & 0.051$\pm$0.001 \\ \hline
\end{tabular}
\caption{Coverage of nominal 95\% posterior intervals calculated from $2000$ datasets simulated under a Poisson model using the same configurations as in Simulation~\upperRomannumeral{3}. The intervals in columns 3 and 4 give the smallest and largest coverage observed for the corresponding parameter. The last two columns give the lengths of nominal 95\% intervals in the format: mean $\pm$ standard deviation.}
\label{table:coveragecompare}
\end{table}

As illustrated in Table~\ref{table:coveragecompare}, when $N$ is increased from $10$ to $40$, the coverage of $\boldsymbol{G}$ changes little. The coverage of $\boldsymbol{B}$ on the other hand increases noticeably even with narrower intervals, especially under the log-$t$ model. The narrowing of the intervals for $\boldsymbol{G}$ is expected with more instruments per source.  The simultaneous increase of coverage and decrease of interval widths for $\boldsymbol{B}$ is intriguing. It is a welcome finding from the astrophysics application perspective. But it also indicates potential defects in the log-Normal or the log-$t$ approximation because over-coverage suggests a non-optimal posterior uncertainty calibration. The half-variance correction likely plays a role here because it permits uncertainty in variance estimation to directly affect inference for the mean. Overall, we recommend  the log-$t$ model when one suspects serious outliers. This may lead to unnecessarily larger error bars for flux estimates that are not (directly) affected by the outliers, a worthwhile premium against disastrous loss of coverages for estimands that are affected.

\section{Applying the Proposed Methods to IACHEC Data}
\label{section:realdataresults}

In this section, we fit the log-Normal model to three datasets (given in Appendix~\ref{appendix:tablesofdata}) compiled by researchers from \citet{IACHEC}, with the aim of increasing understanding of the calibration properties of various X-ray telescopes (a.k.a. instruments) such as \textit{Chandra}, \textit{XMM-Newton}, \textit{Suzaku}, \textit{Swift}, etc. See~\citet{herman2017} for details on data collection and preprocessing.

\subsection{E0102 Data}
\label{section:e0102data}

SNR 1E 0102.2-7219 (abbreviated as E0102) is the remnant of a supernova that exploded in a neighboring galaxy known as the Small Magellanic Cloud~\citep{E0102} and is a calibration target for a variety of X-ray missions. We consider four photon sources associated with E0102. Each is a local peak or ``line'' in the E0102 spectrum, which can be thought of as a high-resolution histogram of the energies of photons originating from E0102. Our ``sources'' corresponds to the photon counts in four bins of this histogram. Two of the lines are associated with highly ionized Oxygen (Hydrogen Lyman-$\alpha$ like O\,VIII at 18.969\AA\ and the resonance line of O\,VII from the He-like triplet at 21.805\AA) and the other two are associated with Neon (H-like Ne\,X at 12.135\AA\ and He-like resonance line Ne\,IX at 13.447\AA). We consider replicate data obtained with $13$ different detector configurations respectively over 4 separate telescopes, {\sl Chandra} (HETG and ACIS-S), XMM-{\sl Newton} (RGS, EPIC-MOS, EPIC-pn), {\sl Suzaku} (XIS), and {\sl Swift} (XRT).

Because the energies of the two Oxygen lines are similar, it is reasonable to assume that their associated Effective Areas are also similar; likewise for the Neon lines and their Effective Areas. Thus, we consider two separate datasets, one with O\,VII and O\,VIII  and the other with Ne\,IX and Ne\,X, each with $M=2$ and $N=13$. In this way, we have more confidence in the multiplicative model (\ref{eqn:multiplicative}) than if we were to combine the two into a single dataset with $M=4$ and $N=13$. In addition, 
standard astronomical practice is to work with dimensionless
measurements in log space (e.g., optical magnitudes in a given passband are defined as $-2.5\log_{10} \frac{\rm flux}{\rm flux~from~Vega}$).  Since our log transformation of brightness mimics this process, we also normalize the measured line fluxes by those from an arbitrarily chosen detector, as done in \cite{plucinsky2017snr}.

To apply the log-Normal model to the two datasets, we choose priors with 
hyperparameters $\alpha=1.5$, $\beta = 2\times 10^{-4}$ for O\,VII, O\,VIII  and $\beta=8\times 10^{-5}$ for Ne\,IX, Ne\,X. We set each $b_i=0$, i.e., a priori we expect no adjustment is needed, with confidence $\tau_i$, taking two possible values $\tau_i=0.025$ and $\tau_i=0.05$. These (and subsequent) choices are based on astronomers' knowledge. 
\begin{figure}[t]
\centering
\includegraphics[width = 0.95\textwidth]{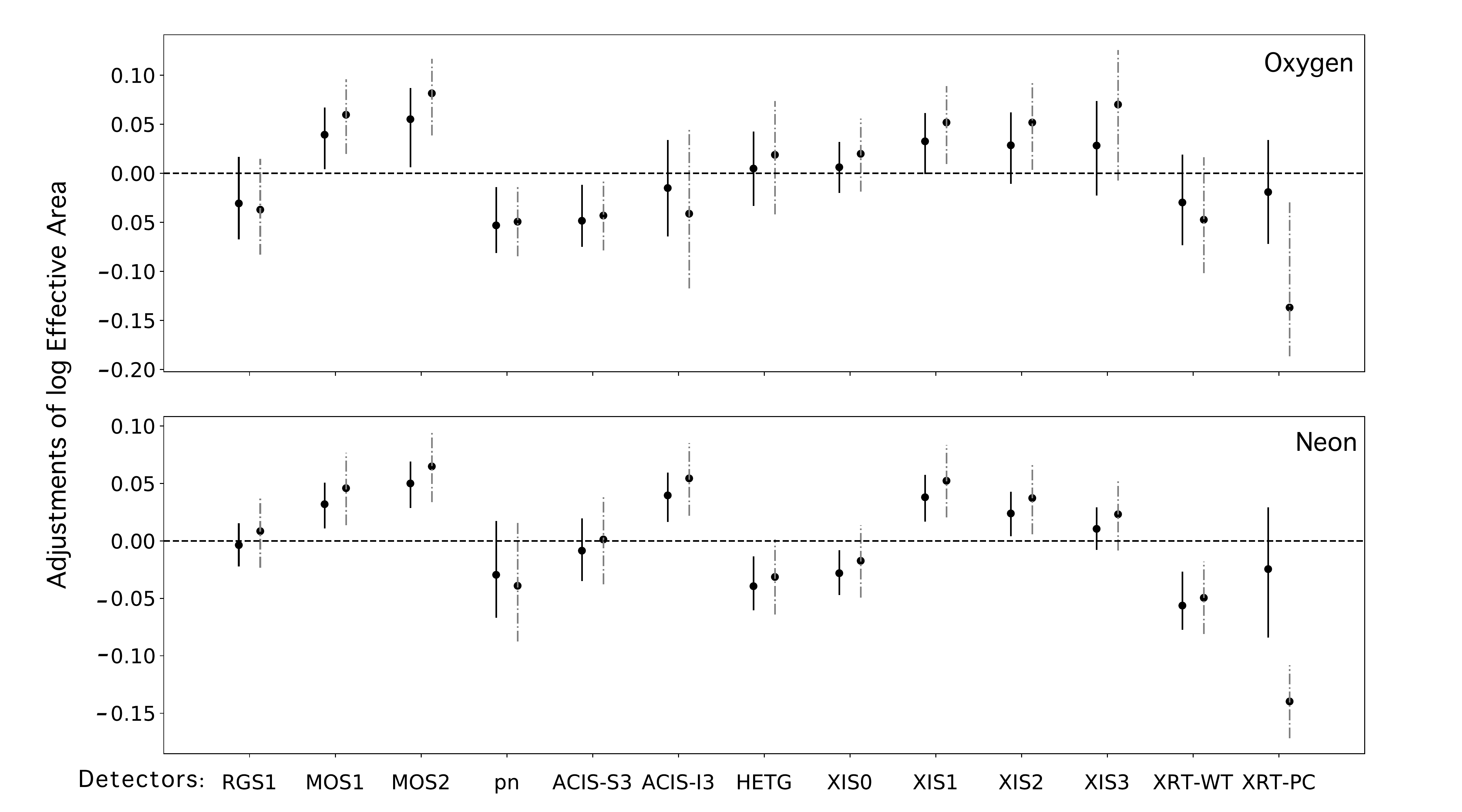}
\caption{Adjustments of the logarithm of the Effective Areas for Oxygen (row 1) and Neon (row 2) in the E0102 dataset. The x-axis labels the detectors (instruments) and the y-axis is $\boldsymbol B$. The horizontal dashed lines represent zero, which indicates no adjustments for the Effective Areas. The vertical bars denote $95\%$ posterior interval for each $B_i$, and the dots denote the posterior means. The black and gray bars correspond to $\tau_i=0.025$ and $0.05$, respectively.}
\label{crossbandNEO}
\end{figure}

Figure~\ref{crossbandNEO} shows the adjustments of the log-scale Effective Areas for Oxygen (row 1) and Neon (row 2) data. We see that the estimated values of $B_i$ are not sensitive to the choices of $\tau_i$ except for detector XRT-PC. 
For XRT-PC, with the Neon data, the estimated shrinkage factor towards the prior, $1-W_i$, as given in (\ref{eq:weight}), is $0.91$ with $\tau_i=0.025$ and $0.02$ with $\tau_i=0.05$. This indicates that, if the prior variance of $B_i$ is too small ($\tau_i=0.025$ here), the model treats the observations as being less accurate (by fitting a large $\sigma_i$) instead of further adjusting the Effective Area of the corresponding instrument (a larger deviation from $b_i$). Numerical results presented in Table~\ref{table:priorinfluence} in  Appendix~\ref{appendix:tablesofdata} reveal that the estimated shrinkage factors can vary slightly or drastically with $\tau_i$. 
Since $M =2$, sensitivity to the choice of hyperparameters is expected. A feature of the log-Normal model is the direct link between its  mean and variance stemming from the half-variance correction. This link indicates additional sensitivities that are neither commonly observed nor well studied.

From Figure~\ref{crossbandNEO} and Table~\ref{table:priorinfluence} (Appendix~\ref{appendix:tablesofdata}), XRT-PC has a much lower Effective Area than the other instruments: about $ -0.15$ versus between $[-0.05, 0.05]$ on the log scale. The corresponding estimated shrinkage factor is more sensitive to the choice of $\tau_i$, for both the Oxygen and Neon data: when $\tau_i$ is small ($=0.025$), the posterior mean of $B_i$ is constrained too much by its zero-centered prior. Thus we need a larger $\sigma_i$ to compensate for the the influence of the prior.
In contrast,  when $\tau_i$ is larger (we also tried $0.05, 0.075, 0.1$), the estimated shrinkage factor is not as sensitive to  $\tau_i$. Overall, Figure~\ref{crossbandNEO} suggests that the Effective Areas of MOS1, MOS2, XIS1, XIS2, XIS3 need to be adjusted upward and those of 
pn, XRT-WT, XRT-PC need to be adjusted downward.

\subsection{2XMM Data}
\label{section:2xmmhmsdata}

The 2XMM catalog~\citep{watson09} can be used to generate large, well-defined samples of various types of astrophysical objects, notably active galaxies (AGN), clusters of galaxies, interacting compact binaries, and active stellar coronae, using the power of {X-ray selection}~\citep{XMM}. The 2XMM catalog data are collected with the XMM-Newton European Photon Imaging Cameras (EPIC). Briefly, there are three EPIC instruments:
the EPIC-pn (hereafter referenced as ``pn'') and the two EPIC-MOS detectors
(hereafter referenced as ``MOS1'' and ``MOS2'').
These detectors have separate X-ray focusing optics but are co-aligned
so that the sources in our samples are observed simultaneously in the pn, MOS1,
and MOS2 detectors.

Our 2XMM data contain three datasets, corresponding to the hard (2.5 - 10.0 keV), medium (1.5 - 2.5 keV) and soft (0.5 - 1.5 keV) energy bands. The three instruments (pn, MOS1 and MOS2) measured 41, 41, and 42  sources respectively in  hard, medium, and soft bands. The sources are from the 2XMM EPIC
Serendipitous Source Catalog \citep{watson09}, selected to be sufficiently
faint that  the thorny issue of ``pileup'', which occurs when several photons hit the detector at the same time, can be ignored. With sufficient exposure, on average 1,500 counts are collected from the faint sources in each band of the detector~\citep{herman2017}. 

\begin{figure}[t]
\centering
\includegraphics[width = 0.95\textwidth]{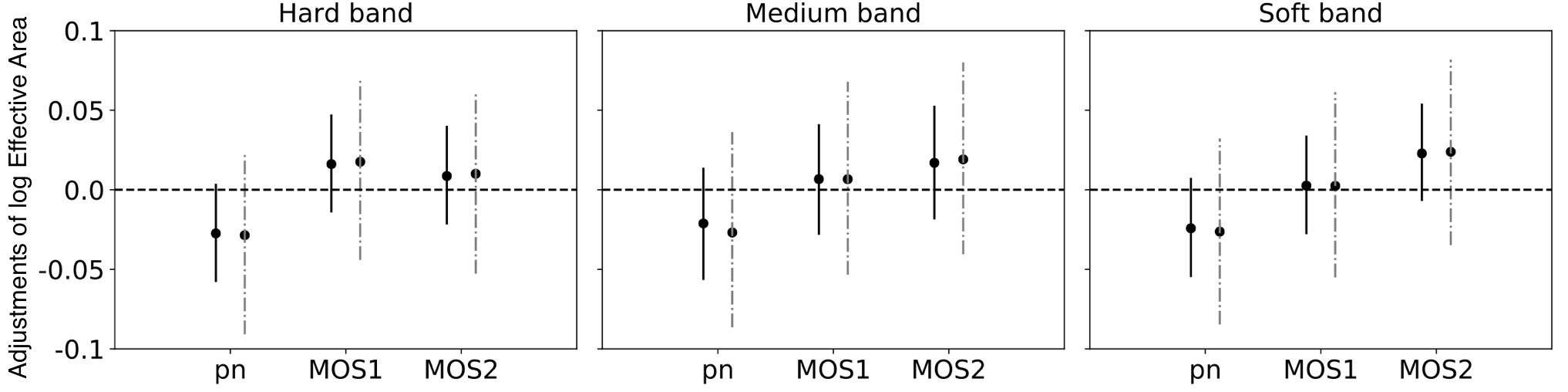}
\caption{Adjustments of the log-scale Effective Areas for hard band (left), medium band (middle) and soft band (right) of the 2XMM datasets. The legend is the same as in Figure~\ref{crossbandNEO}.} 
\label{crossband2XMM}
\end{figure}

The log-Normal model is fit to the three datasets separately, with $\beta = 0.014, 0.083$ and $0.022$ 
respectively for the hard, median, and soft bands, but with $\alpha=1.5$ for all three. We again use $b_i=0$ and try $\tau_i=0.025$ and $\tau_i = 0.05$. Figure~\ref{crossband2XMM} shows the resulting adjustments of the log-scale Effective Area, and confirms the astronomers' expectation that no adjustment is needed for 2XMM, regardless of the choice of the $\tau_i$. In contrast to Table~\ref{table:priorinfluence}, Table~\ref{table:priorinfluencehardbandetc} (also in Appendix A) shows a much more stable patterns of proportion of prior information for 2XMM data.

\subsection{XCAL Data}
\label{section:hmsdata}

XCAL consists of bright AGN from the XMM-Newton cross-calibration sample\footnote{See Section 4 in \url{http://xmm2.esac.esa.int/docs/documents/CAL-TN-0052.ps.gz}}. The image data are clipped, using a standard XMM software task (called {\tt epatplot}),
to eliminate the regions affected non-trivially by pileup. The amount of clipping depends
on the observed source intensity: unused regions are larger for brighter
sources \citep{herman2017}. The initial estimate of the Effective Area is then adjusted according to lookup tables (from other in-flight data) to account for the unused regions.  Like the 2XMM data, XCAL data are composed of three datasets: the hard (94 sources), medium (103), and soft (108) bands, all measured by three instruments, pn, MOS1 and MOS2. We use the same procedure and hyperparameters as in Section~\ref{section:2xmmhmsdata}, except we set $\beta = 8.0\times 10^{-4},8.6\times 10^{-3}$ and $6.8\times 10^{-4}$ respectively for hard, median, and soft bands.

Figure~\ref{medbandG} demonstrates that adjustment of the Effective Areas is needed to align the measured fluxes across the detectors. Results are presented for four sources from the medium band data, where the left three bars---corresponding to three instruments---depict the $95\%$ intervals (mean $\pm$ 2 given standard deviations) for the log-fluxes obtained by a standard astronomical method. The right two bars---corresponding to two choices of the prior variance $\tau_i$---represent the $95\%$ posterior intervals of log-fluxes after adjustment using our log-Normal model. This visualization illustrates the reliability of our calibration of Effective Areas, as it helps to bring together the varied flux estimates from individual detectors in a statistically principled way. In particular, we see that the posterior mean of the log-flux is rather robust to the choice of $\tau_i$, yet the corresponding posterior variance respects astronomers' a priori knowledge as coded into $\tau_i$.  

\begin{figure}[t]
\centering
\includegraphics[width = 0.8\textwidth]{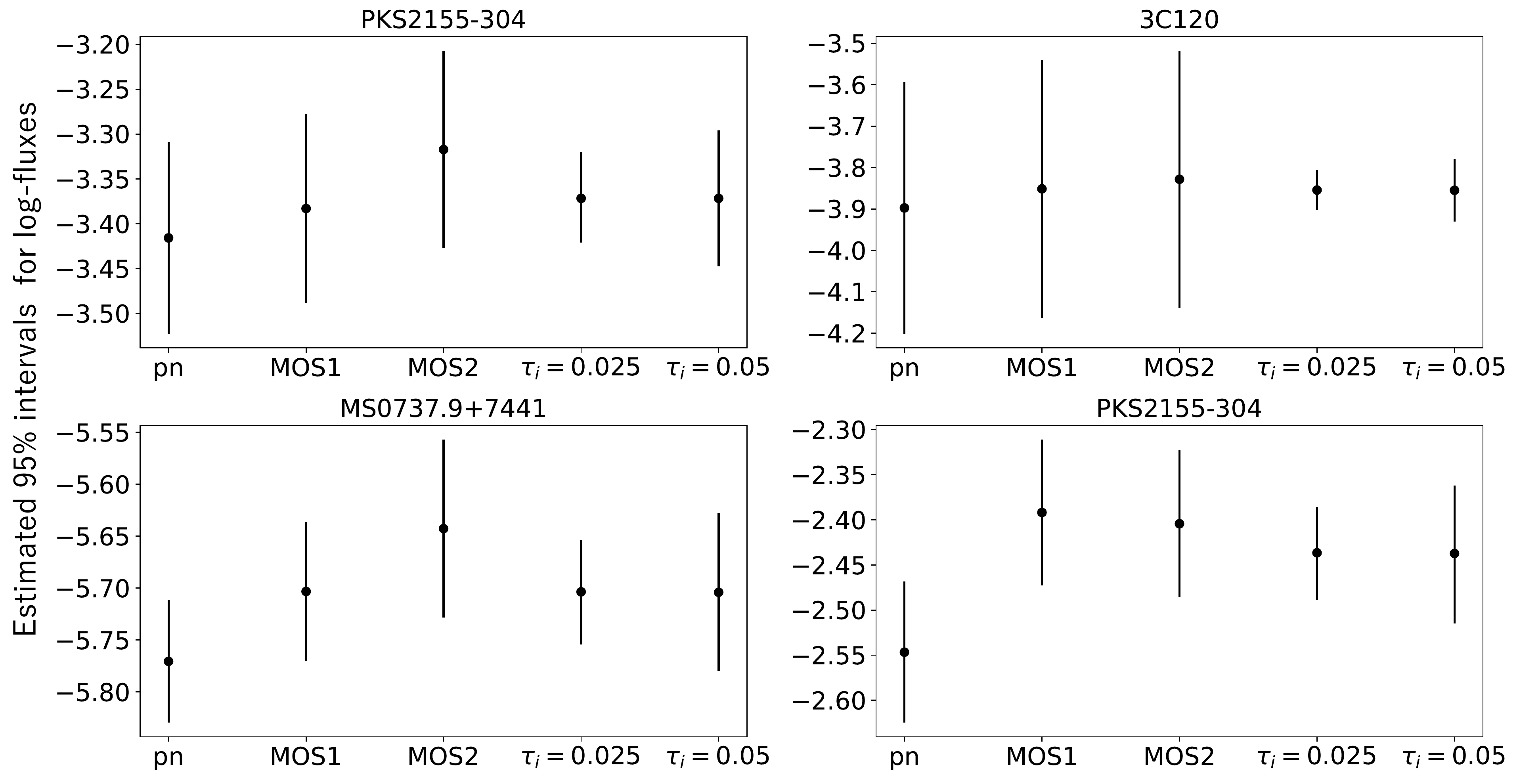}
\caption{Comparison of estimated $95\%$ intervals for log-fluxes using a standard astronomical method (left three bars) and those from the log-Normal model (right two bars) for four representative medium-band sources in XCAL data, as indicated by the panel titles.}
\label{medbandG}
\end{figure}

Finally, we show how to adjust the Effective Areas of each instrument to obtain the results illustrated in the rightmost interval in each panel of Figure~\ref{medbandG}. Figure~\ref{crossbandv5} shows the necessary adjustment of $\boldsymbol B$ for hard band (left), medium band (middle) and soft band (right). For all these bands, we adjust the Effective Area of pn downward and that of MOS2 upward.

\begin{figure}[t]
\centering
\includegraphics[width = 0.95\textwidth]{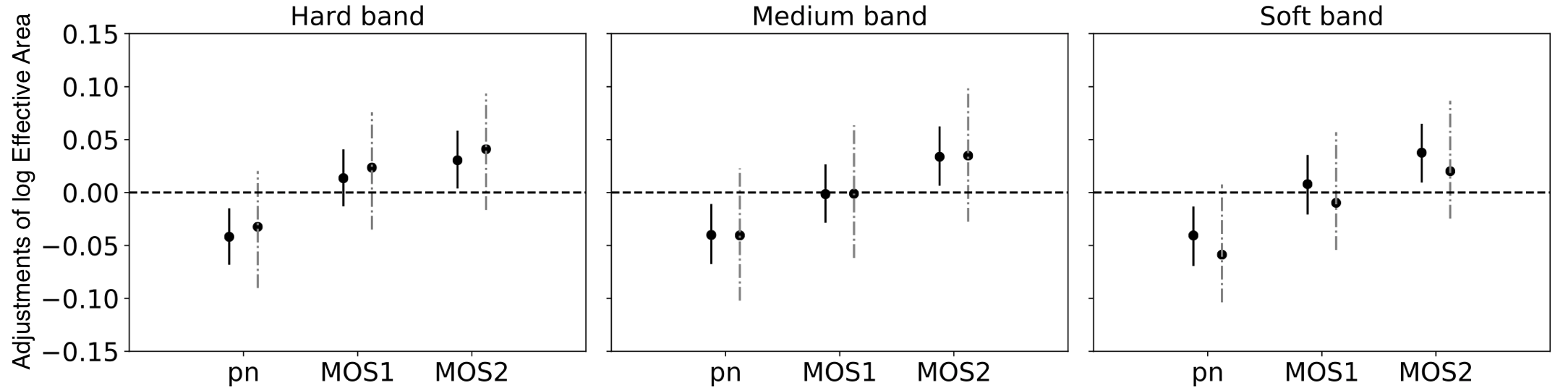}
\caption{Adjustments of the log-scale Effective Areas for hard band (left), medium band (middle) and soft band (right) based on XCAL data. The legend is the same as in Figure~\ref{crossbandNEO}.}
\label{crossbandv5}
\end{figure}

\subsection{Model Checking}
\label{section:goodnessoffit}

To check how well the log-Normal model captures the observed variability in the data, we use residual plots to visualize the goodness-of-fit. Figure~\ref{fig:residueexample} plots the standardized residuals $\widehat{\mathcal{R}}_{ij}$ 
for the data analyzed in Section~\ref{section:hmsdata} with $\tau_i=0.05$, with the left panels denoting residuals from the log-Normal model and the right panels from the log-$t$ model (see the two expressions of $\widehat{\mathcal{R}}_{ij}$ in~\eqref{eqn:standardizedresidue}). Nearly all residuals fall in $[-3,\, 3]$ under the log-Normal model and $[-2,\, 2]$ under the log-$t$ model. The observations of \texttt{3C111} in all three energy bands are the only outliers under the log-Normal model but are not outliers under the log-$t$ model, confirming the latter's ability to handle outliers. The adjusted Effective Areas and the estimated fluxes are not too sensitive to whether or not the outliers are excluded. Thus the log-Normal model is acceptable for the data in Section~\ref{section:hmsdata}. 

\begin{figure}[t]
\centering
\includegraphics[width = 0.99\textwidth]{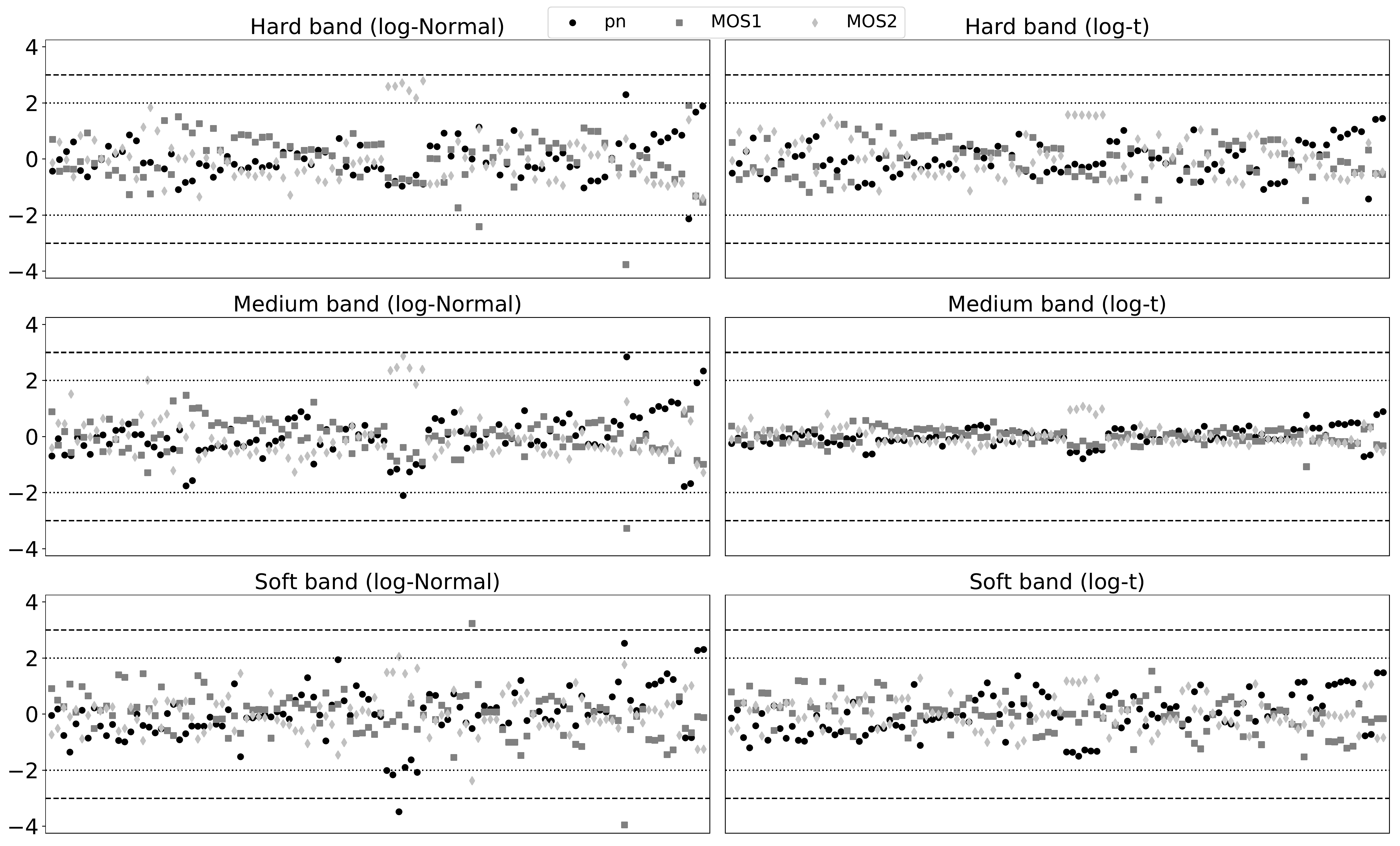}
\caption{Standardized residuals for the XCAL data in Section~\ref{section:hmsdata} with $\tau_i=0.05$. Left panel for the log-Normal model and right panel for the log-$t$ model. The black circles, gray squares and silver rhombi denote the instruments pn, MOS1 and MOS2 respectively. The dashed and dotted lines depict respectively the $[-3,\, 3]$ and $[-2, 2]$ intervals.}
\label{fig:residueexample}
\end{figure}

We also employ a posterior predictive check~\citep{meng1994posterior, gelman1996posterior} for the log-Normal model. In a posterior predictive check, one chooses test statistics and computes the posterior predictive $p$-value. The test statistics we choose are $$\left\{T_i = \overline{y}_{i\cdot} -\overline{y} = \frac{\sum_{j=1}^M y_{ij}}{M} - \frac{\sum_{i=1}^N\sum_{j=1}^M y_{ij}}{NM}\right\}_{i=1}^{N},$$ which reflect the relative magnitudes of the log scale Effective Areas. None of the posterior predictive $p$-values for any of our datasets are significant, i.e., we never fail the posterior predictive check. This does not, however, imply that no serious model defects exist. Below we  discuss directions for improving our models and ultimately the reliability of the proposed concordance adjustments.

\section{Alternative Methods and Future Work} 
\label{sec:discussion}
\subsection{Comparing Likelihood and Bayesian Estimations}
\label{section:alternative}

In Section~\ref{section:calibrationmodel}, we adopt a Bayesian perspective, which leads to the log-Normal model. Alternatively, we could view $b_i$ as a \textit{noisy observation} with known variance $\tau_i^2$: ${b}_i {\rm \stackrel{indep}{\sim}}  \mathcal{N}({B}_i, \tau_i^2)$. Together with (\ref{eq:threey}), this gives a multivariate Normal regression model and can be fit via maximum likelihood estimation (MLE). In particular, when $\boldsymbol{\sigma^2}$ is known, Proposition~\ref{proposition:variance_BG_known_var} in Appendix~\ref{appendix:mlecompute} gives closed-form expressions when all the instruments measure all sources, which implies the asymptotic properties of the MLEs (Corollary~\ref{coro:asymptoticvarianceMLE}). Furthermore, the standardized residual sum of squares follows a $\chi^2$ distribution, which enables testing of the goodness-of-fit; see Appendix~\ref{subsec:goodnessoffitregression} for details.

When the variances $\boldsymbol{\sigma^2}$ are unknown, in principle, we can still obtain the (asymptotic) variance of the MLEs by calculating the observed/expected Fisher information. However, the number of unknown parameters we consider, $2N+M$, grows with the number of observations $NM + N$. Conditions on the relationship between $N$ and $M$ that would ensure the classical asymptotic theory for MLEs would be of interest to those who prefer likelihood inference. Furthermore, under such conditions, these  estimators should be \textit{approximately} valid even if the Normal assumptions made in (\ref{eq:threey}) for the $e_{ij}$ fail. In this case, the variance of the estimator requires a more complicated ``sandwich" formula~\citep{freedman2006so}, involving both the Fisher information and the variance of the score function. We say \textit{approximately} valid because the half-variance correction of Section~\ref{section:calibrationmodel} would depend on the log-Normal assumption. Consequently, when the variance is large, the half-variance correction can be misleading  if the log-Normal assumption is severely violated. 

As usual, the likelihood method is closely related to the Bayesian approach. For example, when  $\boldsymbol{\sigma^2}$ is known, the MLEs of $\boldsymbol{B}$ and $\boldsymbol{G}$ correspond to the MAP estimates defined in (\ref{eqn:shrinkBG}), which also have the intuitive interpretation as shrinkage estimators. When the variances are unknown, the likelihood function is unbounded at the boundary of the parameter space ($\sigma_i^2=0$). The conjugate priors for variance parameters in the Bayesian model regularize the likelihood and give a proper posterior distribution. This is another reason we adopt the Bayesian approach.

\subsection{Future and Related Work}
\label{sec:conclusion}

The log-Normal model works reasonably well in our applied illustrations,  and it yields important findings that are welcomed by astronomers --- concrete guidance about systematic adjustments of the Effective Areas are given and thus concordance of an intrinsic characteristic of each astronomical object across different instruments can be achieved. Calibration scientists are thus able to make \textit{absolute} measurements of characteristics of astronomical objects using different instruments. The posterior distribution of the Effective Area of each instrument can and should be used for downstream analysis of measured fluxes to obtain principled estimates of absolute flux and to properly quantify their uncertainty. Furthermore, we highlight the danger of incorrectly fixing the observation noise through simulation experiments that mimic possible realistic uncertainties.

There are several directions of future work in order to improve the current model. First, so far we have assumed that the Effective Areas are a priori independent, which is not always true in practice. Sometimes the Effective Areas across different energy bands are noticeably correlated. This correlation structure should be taken into account in future modeling to gain more efficiency in estimation. Second, the log-Normal model gives conservative results under realistic model misspecification, as revealed by our simulation studies. Unfortunately, the scope of simulation studies is always limited. Hence theoretical properties of both the log-Normal and the log-$t$ approaches under model misspecification need to be further investigated. Third, the asymptotic (sampling) properties of the proposed models need to be established, as both the number of parameters and the number of observations approach infinity. Fourth, the robustness of the model with respect to possible sample selection bias and non-ignorable missing data needs to be studied more thoroughly. Although this is not of concern for our current analyses, it could become a more severe problem if we include more instruments and more sources in the calibration data. Last, possible hierarchical extensions of the model that addresses population characteristics, in which we are interested in the flux distribution of a certain type or population of objects instead of each individual object, as in the case of supernovae, could be considered. Of course, in such cases, a representative sample of the population of objects is critical for a meaningful analysis.

Moving forward, to increase the impact of the proposed method, we need to involve more IACHEC members and datasets.  Cooperation among IACHEC member projects can lead to enacting adjustments as recommended from the concordance analysis, which will result in closer agreement between different instruments that make similar measurements, to achieve a main goal of IACHEC.  Experts from the projects that comprise the IACHEC are needed to examine possible bias in sample selection and to set the values of $\boldsymbol\tau$ that are needed for the concordance analysis.  In our follow-up paper \citep{herman2017}, we apply this concordance analysis more broadly and allow the values of $\boldsymbol\tau$ to be instrument-dependent.

Finally, calibration is a well-known problem in several areas of applications. For example, inter-laboratory calibration (motivated from analytical chemistry) is studied in \citet{gibbons2001weighted} and \citet{bhaumik2005confidence}, where they also address simultaneously the issue of multiplicative signals and additive noises, but with a different modeling strategy.  More recently, a fiducial approach is used in~\citet{hannig2017fusion} to tackle similar problems. We therefore hope our modeling strategies add to the toolkits to conduct similar calibration and concordance analysis, such as for environmental monitoring  \citep[e.g.,][]{weatherhead1998factors}. Much more  work is needed and can be done, and hence we invite and encourage interested researchers to join us to address these theoretically challenging and practically impactful problems.

\section*{Acknowledgement}

This project was conducted under the auspices of the CHASC International Astrostatistics Center. CHASC is supported by the NSF grants DMS-15-13484, DMS-15-13492, DMS-15-13546, DMS-18-11308, DMS-18-11083, and DMS-18-11661. In addition, David van Dyk's work was supported by a Marie-Skodowska-Curie RISE (H2020-MSCA-RISE-2015-691164) Grant provided by the European Commission. Vinay Kashyap and Herman Marshall acknowledge support under NASA Contract NAS8-03060 with the {\sl Chandra} X-ray Center.  We also thank Matteo Guainazzi, Paul Plucinsky, Jeremy Drake, Aneta Siemiginowska, and other members of the IACHEC and CHASC collaborations for valuable discussions.

\baselineskip=10pt
\bibliographystyle{apalike}
 
\bibliography{calibrationbib}

\spacingset{1.45}

\appendix

\section{Tables of Data Description and Prior Influence}
\label{appendix:tablesofdata}

Tables~\ref{table:e0102} and~\ref{table:2XMM_and_XCAL} give summaries of the data used in Sections~\ref{section:e0102data},~\ref{section:2xmmhmsdata} \&~\ref{section:hmsdata}.

\begin{table}[tbph]
\centering
\begin{tabular}{c|cc|cc}
\hline
Lines (Sources) & He-like OVII & H-like OVIII & He-like Ne\,IX & H-like Ne\,X\\
Spectrum & 21.805\AA & 18.969\AA & 13.447\AA & 12.135\AA\\
\hline
\end{tabular}\\
\vspace{0.2in}
\begin{tabular}{c|cccc}
\hline
Telescopes & {\sl Chandra} & XMM-{\sl Newton} & {\sl Suzaku} & {\sl Swift}\\
Detectors (Instruments) & HETG, ACIS-S & RGS, EPIC-MOS, EPIC-pn & XIS & XRT \\
\hline
\end{tabular}
\caption{Summary of E0102 data. The first table gives the sources for two data sets, Highly ionized Oxygen and Neon. The second table gives instruments for both data sets.}
\label{table:e0102}
\end{table}

\begin{table}[tbph]
\centering
\begin{tabular}{c|ccc|ccc}
\hline
Observatory & \multicolumn{6}{|c}{XMM-{\sl Newton} European Photon Imaging Cameras (EPIC)}\\
\hline
Detectors (Instruments) & \multicolumn{6}{c}{EPIC-pn (pn), EPIC-MOS (MOS1 \& MOS2)} \\
\hline
Data Acronym & \multicolumn{3}{|c|}{2XMM} & \multicolumn{3}{|c}{XCAL}\\
\hline
Energy Band & Hard & Medium & Soft & Hard & Medium & Soft\\
Energy (keV) & 2.5-10.0 & 1.5-2.5 & 0.5-1.5 & 2.5-10.0 & 1.5-2.5 & 0.5-1.5\\
\hline
No. Sources & 41 & 41 & 42 & 94 & 103 & 108 \\
\hline
\end{tabular}
\caption{Summary of 2XMM data and XCAL data. The number of instruments is $N=3$ (pn, MOS1, MOS2) and the number of sources ($M$) is given in the last row for the six data sets, three from different energy bands of 2XMM data and XCAL data respectively.}
\label{table:2XMM_and_XCAL}
\end{table}

\newpage
\begin{table}[tbph]
\centering
\begin{tabular}{c|cc|cc} 
Instrument & \multicolumn{2}{|c|}{Oxygen} & \multicolumn{2}{|c}{Neon} \\
& $\tau = 0.025$ & $\tau = 0.05$  & $\tau = 0.025$ & $\tau = 0.05$  \\
\hline
RGS1 & 0.570 & 0.205 & 0.063 & 0.016 \\
MOS1 & 0.279 & 0.077 & 0.075 & 0.019 \\
MOS2 & 0.355 & 0.065 & 0.077 & 0.017 \\
pn & 0.250 & 0.041 & 0.620 & 0.218 \\
ACIS-S3 & 0.218 & 0.040 & 0.270 & 0.088 \\
ACIS-I3 & 0.906 & 0.640 & 0.099 & 0.026 \\
HETG & 0.648 & 0.341 & 0.129 & 0.034 \\
XIS0 & 0.180 & 0.051 & 0.069 & 0.018 \\
XIS1 & 0.298 & 0.078 & 0.071 & 0.019 \\
XIS2 & 0.463 & 0.140 & 0.063 & 0.016 \\
XIS3 & 0.772 & 0.364 & 0.062 & 0.018 \\
XRT-WT & 0.726 & 0.278 & 0.154 & 0.026 \\
XRT-PC & 0.934 & 0.235 & 0.906 & 0.017 \\
\hline
\end{tabular}
\caption{Proportion of prior influence, as defined by $1-W_i$ (of (\ref{eq:weight})), for E0102 data in Section~\ref{section:e0102data}.}
\label{table:priorinfluence}
\end{table}

\begin{table}[tbph]
\centering
\begin{tabular}{c|ccc|ccc}
Data Name  & \multicolumn{3}{|c|}{$\tau_i=0.025$} & \multicolumn{3}{|c}{$\tau_i=0.05$} \\
&  pn & mos1 & mos2 & pn & mos1 & mos2 \\
\hline
hard band 2XMM  & 0.093 & 0.075 & 0.082 & 0.025 & 0.020 & 0.022\\
medium band 2XMM  & 0.250 & 0.216 & 0.222 & 0.076 & 0.065 & 0.067\\
soft band 2XMM  & 0.093 & 0.075 & 0.069 & 0.025 & 0.020 & 0.018\\
hard band XCAL & 0.010 & 0.019 & 0.031 & 0.003 & 0.005 & 0.008 \\
medium band XCAL & 0.023 & 0.016 & 0.028 & 0.006 & 0.004 & 0.007
\\
soft band XCAL & 0.021 & 0.011 & 0.007 & 0.005 & 0.003 & 0.002\\
\hline
\end{tabular}
\caption{Proportion of prior influence for data used in the analysis in Sections~\ref{section:2xmmhmsdata} and~\ref{section:hmsdata}.}
\label{table:priorinfluencehardbandetc}
\end{table}

\newpage

\section{Details of Fitting the Log-Normal Model}
\label{appendix:bayescomputation}

The following three MCMC algorithms are used for our posterior sampling. 

\begin{enumerate}
\item \textbf{Standard Gibbs Sampler:} iterates the following three sets of conditional distributions, all easily derived from (\ref{eqn:joint}):
\begin{enumerate}
\item Conditioning on $\boldsymbol{G}$ and $\boldsymbol{\sigma}^2$, sample $B_i$ independently for $i=1,\ldots, N$  from 
\begin{equation*}
\mathcal{N} \bigg(\frac{b_i/ \tau_i^2 + \sum_{j\in J_i} (y_{ij} +0.5\sigma_{i}^2 - G_j) / \sigma_{i}^2 }{1/\tau_i^2 + \sum_{j\in J_i} 1/\sigma_{i}^2},\  \frac{1}{1/\tau_i^2 + \sum_{j\in J_i} 1/\sigma_{i}^2}\bigg).
\end{equation*}
\item Conditioning on $\boldsymbol{B}$ and $\boldsymbol{\sigma}^2$, sample $G_j$ independently for $1\leq j\leq M$ from 
\begin{equation*}
\mathcal{N} \bigg(\frac{\sum_{i\in I_j} (y_{ij}+0.5\sigma_{i}^2- B_i) / \sigma_{i}^2}{\sum_{i\in I_j} 1/\sigma_{i}^2},\  \frac{1}{\sum_{i\in I_j} 1/\sigma_{i}^2}\bigg).
\end{equation*}
\item Conditioning on $\boldsymbol{B}$ and $\boldsymbol{G}$, sample $\sigma_i^2$ independently for  
$i=1, \ldots, N$ from 
\begin{equation*}
\sigma_{i}^{-|J_i|-2-2 \alpha}\exp\left\{- \frac{1}{2}\frac{\sum_{j\in J_i} (y_{ij}-B_i-G_j)^2+2\beta}{\sigma^{2}_{i}}-\frac{|J_i|\sigma^2_i}{8}\right\}
\end{equation*}
via the Metropolis-Hastings algorithm using a simple random walk proposal (Gaussian proposal) on the log-scale, i.e., $\log (\sigma_i^2)$. \label{stepc}
\end{enumerate}
\item \textbf{Block Gibbs Sampler}: same as above except replace the two conditional steps (1a) and (1b) by a joint draw of $\{\boldsymbol{B},\boldsymbol{G}\}$ from $(N+M)$-dimensional Gaussian with mean $\boldsymbol{\Omega(\sigma^2)}^{-1} \boldsymbol{\gamma(\sigma^2)}$ and covariance matrix 
$\boldsymbol{\Omega(\sigma^2)}^{-1}$; see (\ref{eq:gamma}) and (\ref{eqn:expressionomegasigma})  in Section~\ref{section:hierregressionmodel}. 
\item \textbf{Hamiltonian Monte Carlo (HMC)}: samples the entire vector $\boldsymbol{\theta} = \{B_i, G_j,\sigma_i^2\}$ through the non-U-turn HMC sampler~\citep{NUTS}, implemented with the \texttt{STAN} package. Here we give a brief description of HMC; see \citet{HMC} for more details.  Let $\pi(\boldsymbol{\theta})$ denote the (unnormalized) joint posterior $\boldsymbol{\theta}$, as given by (\ref{eqn:joint}). Define potential energy as $U(\boldsymbol{\theta})=-\log \pi(\boldsymbol{\theta})$ and kinetic energy as $k(\boldsymbol{p}) = \boldsymbol{p}^\top \mathcal{M}^{-1} \boldsymbol{p}$, where $\mathcal{M}$ is a symmetric positive-definite matrix, thus the total energy is $H(\boldsymbol{\theta},\boldsymbol{p}) = U(\boldsymbol{\theta})+k(\boldsymbol{p})$. We can obtain samples of $\pi(\boldsymbol{\theta})$ by sampling from the target density $\exp[-H(\boldsymbol{\theta},\boldsymbol{p})] \propto \pi(\boldsymbol{\theta}) \exp(-\boldsymbol{p}^\top \mathcal{M}^{-1} \boldsymbol{p})$, which is essentially a data-augmentation technique~\citep{tanner1987calculation}. By defining the potential energy and kinetic energy, we can propose MCMC moves according to the Hamiltonian dynamics, which explores the parameter space more efficiently by taking bigger and less correlated moves, as opposed to random walk Metropolis-Hastings or a Gibbs sampler. In practice, we use the leapfrog move to approximate the Hamiltonian dynamics. Due to the energy-preserving property of Hamiltonian dynamics, the acceptance rate of the resulting HMC is approximately $1$. It is not exactly $1$ because we use the (discretized) leapfrog moves to approximate (continuous) Hamiltonian dynamics. The tuning parameters of the HMC algorithm include the covariance matrix $\mathcal{M}$, the leapfrog step size $\epsilon$, and the number of leapfrog steps $L$. These are all self-tuned in the \texttt{STAN} package.
\end{enumerate}

We compare the performance of these three algorithms using auto-correlation plots of the posterior samples and the effective sample size, in both the simulated and real data examples. Not surprisingly, the Gibbs sampler converges very slowly relative to the other two algorithms. We are able to cross check our results by comparing the samples obtained with the block Gibbs sampler and HMC -- they give practically the same posterior distributions.

\section{Proprieties of the Posterior Distribution}

\subsection{Propriety of Posterior}
\label{section:properposterior}
\begin{theorem}
\label{theorempropergeneral}
Under the prior specifications for $\{B_i, G_j,\sigma_i^2: 1\leq i\leq N, 1\leq j\leq M\}$ given in (\ref{eqn:lognormalmodel}), the posterior is proper if each source is measured by at least one instrument, i.e., $|I_j| \geq 1$ for all $1\leq j\leq M$.
\end{theorem}
\begin{proof} We prove the propriety of the posterior by first integrating out the $G_j$ first, then the $B_i$, and finally the $\sigma_i^2$. By (\ref{eqn:joint}), $p(\boldsymbol{B},\boldsymbol{G}, \boldsymbol{\sigma}^2 | \boldsymbol{D}, \boldsymbol{\tau}^2)$ is proportional to
\begin{equation}\label{eqn:joint1}
\prod_{i=1}^N \sigma_{i}^{-|J_i|-2-2 \alpha}\exp\left\{- \frac{1}{2}\sum_{j=1}^M\sum_{i\in I_j} \sigma^{-2}_{i} (y'_{ij}-B_i-G_j)^2 -\sum_{i=1}^N \left[\frac{(b_{i}-B_i)^2}
{2\tau^2_{i}} +\frac{\beta}{\sigma_i^2}\right]\right\}.
\end{equation}
Now for each $1\leq j\leq M$, if we define a random index ${\cal I}$ on $I_j$ such that $\Pr({\cal I}=i)\propto \sigma_i^{-2}$, then  
\begin{equation}
\frac{\sum_{i\in I_j} \sigma^{-2}_{i} (y'_{ij}-B_i-G_j)^2}{\sum_{i\in I_j} \sigma_i^{-2}} =
\E \left[y'_{{\cal I}j}-B_{\cal I}-G_j\right]^2 \ge  
\left[\E (y'_{{\cal I}j}-B_{\cal I})-G_j\right]^2. 
\end{equation}
Therefore, the first term in the exponential part of (\ref{eqn:joint1}) is less than $-0.5 \left(\sum_{i\in I_j}\sigma_i^{-2}\right)(G_j-C_j)^2$, where $C_j=\E (y'_{{\cal I}j}-B_{\cal I})$ is free of $G_j$. The property of Normal density (for $G_j$) then yields
\begin{equation*}
\int p(\boldsymbol{B},\boldsymbol{G},\boldsymbol{\sigma}^2 | \boldsymbol{D}, \boldsymbol{\tau}^2) \ d\boldsymbol{G} \leq C^* \prod_{i=1}^N \sigma_{i}^{-|J_i|-2-2 \alpha} \prod_{j=1}^M \left[ \sum_{i\in I_j} \sigma_i^{-2} \right]^{-1/2} \exp\left\{-\sum_{i=1}^N \left[\frac{(b_{i}-B_i)^2}
{2\tau^2_{i}} +\frac{\beta}{\sigma_i^2}\right]\right\}
\end{equation*}
where $C^{*}$ is a constant that depends only on $\boldsymbol{D,\tau^2}$. Integrating out $\boldsymbol{B}$ then gives
\begin{align}\label{eqn:intB}
\int \int p(\boldsymbol{B},\boldsymbol{G},\boldsymbol{\sigma}^2 | \boldsymbol{D}, \boldsymbol{\tau}^2) \ d\boldsymbol{G}\ d\boldsymbol{B} \leq C^{**} \prod_{i=1}^N \sigma_{i}^{-|J_i|-2-2 \alpha} \prod_{j=1}^M \left[ \sum_{i\in I_j} \sigma_i^{-2} \right]^{-1/2} \exp\left\{-\sum_{i=1}^N \frac{\beta}{\sigma_i^2}\right\}
\end{align}
where $C^{**}$ is a constant that depends only on $\boldsymbol{D,\tau^2}$. Since $I_j$ is non-empty, it is meaningful to invoke the well-known harmonic-geometric mean inequality to obtain that
\begin{equation}\label{eqn:hg}
\prod_{j=1}^{M}\left[ \sum_{i\in I_j} \sigma_i^{-2} \right]^{-1/2} \leq\prod_{j=1}^M |I_j|^{-1/2}\left[ \prod_{i\in I_j} \sigma_i \right]^{1/|I_j|} \leq
\prod_{i=1}^N \sigma_i^{\sum_{j\in J_i} |I_j|^{-1}}. 
\end{equation}
Inequalities (\ref{eqn:intB}) and (\ref{eqn:hg}) together imply that the unnormalized $p(\boldsymbol{\sigma}^2 | \boldsymbol{D}, \boldsymbol{\tau}^2)$ is dominated above by a constant times $\prod_{i=1}^N p_i(\sigma_i^2)$, where $p_i(x)$ is the density of the inverse Gamma distribution with shape parameter $\alpha_i = \alpha+ [|J_i|- \sum_{j\in J_i} |I_j|^{-1}]/2$ and scale parameter $\beta$. Because $|I_j|\ge 1$, we have $\alpha_i\ge \alpha$. Hence as long as the hyperparameter $\alpha>0$, which is always chosen to be so, $p_i$ is a proper density. Consequently, $p(\boldsymbol{\sigma}^2 | \boldsymbol{D}, \boldsymbol{\tau}^2)$ is a proper density after renormalization.

\end{proof}

\subsection{Identifiability}
\label{appendix:identifiability}

When $\tau_i^2$ is large, the likelihood information for estimating $B_i$ (i.e., from $c_{ij}$) dominates the prior information (i.e., from $b_i$). In the extreme case of $\tau_i^2 = \infty$, the model is not identifiable because for fixed variances, $\{{B}_i, {G}_j\}$ and $\{B_i + \delta, {G}_j -\delta\}$ yield the same posterior densities for $\{\boldsymbol{B},\boldsymbol{G}\}$ for any constant $\delta$. Let $\lambda_{\rm{max}}$ and $\lambda_{\rm{min}}$ be the maximum and minimum eigenvalues of $\boldsymbol{\Omega(\sigma^2)}$,  as defined in Section~\ref{section:hierregressionmodel}. Taking $u = (\mathbf{1}_N,\mathbf{1}_M)^{\top}$ and $v = (\mathbf{1}_N,-\mathbf{1}_M)^{\top}$, the \textit{condition number} of $\boldsymbol{\Omega(\sigma^2)}$ is
\begin{equation}
\frac{\lambda_{\rm{max}}}{\lambda_{\rm{min}}}\geq \frac{u^{\top}\boldsymbol{\Omega(\sigma^2)}\ u}{v^{\top}\boldsymbol{\Omega(\sigma^2)} \ v} =1+\frac{4\sum_{i=1}^N|J_i|\sigma^{-2}_{i}}{\sum_{i=1}^N\tau_i^{-2}},
\label{eqn:eigenvalue}
\end{equation}
where $\boldsymbol{1}_n$ denotes an $n\times 1$ vector of ones. As a consequence, when $\{\tau_i^2\}$ are generally larger than $\{\sigma_{i}^2\}$, the ratio in (\ref{eqn:eigenvalue}) can be large, and the posterior contours, determined by $\boldsymbol{\Omega} $, are elongated in one direction and narrow in another. This provides a guideline  that $\{\tau^2_i\}$ should not be set too large relative to $\{\sigma_i^2\}$ in practice, because large $\{\tau^2_i\}$ can lead to near model non-identifiability and consequently more costly computation.  A computationally cheaper way of dealing with possible model non-identifiability is to set one of the $\{B_i\}$ equal to a fixed value, which is equivalent to setting the corresponding $\tau_i=0$.  We experiment with this computationally cheap strategy in our  empirical evaluations, and find that it does not alter the results in substantive ways, but the resulting estimators for the Effective Areas are relative to some (arbitrarily) chosen values instead of in absolute terms/magnitudes. 

\section{Derivation of Conditional Covariance Matrix}
\label{appendix:derivationofomegainverse}

In this section, we give detailed derivations of $\boldsymbol{\Omega}^{-1}(\boldsymbol{\sigma}^2)$ when all instruments measure all sources. In this case, $W_i$ defined in ~\eqref{eq:weight} becomes $W_i = \frac{M \sigma_i^{-2}}{M\sigma_i^{-2} +\tau_{i}^{-2}}$, $1\leq i\leq N$. Define $\tilde{\sigma}^2 = \left(N^{-1}\sum_{i=1}^N \sigma_i^{-2}\right)^{-1}$.  

Let $\boldsymbol{A}$ be the $(N+M)\times (N+M)$ diagonal matrix with diagonal elements equal to those of $\boldsymbol{\Omega}(\sigma^2)$. Let $\boldsymbol{U}$ be an $(N+M)\times 2$ matrix such that $U_{i,1}= \sigma_i^{-2},\ U_{i,2} = 0 $ for $i=1, \ldots, N$, and $U_{j+N,1} = 0,\ U_{j+N,2}=1$ for $j=1,\ldots,M$. Let $\boldsymbol{C}$ be a $2\times 2$ matrix such that $C_{i,j}=I_{i\neq j}\ (i,j=1,2)$.
Then $\boldsymbol{\Omega}(\boldsymbol{\sigma}^2) = \boldsymbol{A} + \boldsymbol{U}\boldsymbol{C} \boldsymbol{U}^\top$. By the Woodbury matrix identity, we have
\begin{equation}
\label{equ:inversematrix}
\boldsymbol{\Omega}^{-1} (\boldsymbol{\sigma}^2) = \boldsymbol{A}^{-1} - \boldsymbol{A}^{-1}\boldsymbol{U}\left(\boldsymbol{C}+\boldsymbol{U}^\top\boldsymbol{A}^{-1}\boldsymbol{U}\right)^{-1}\boldsymbol{U}^\top\boldsymbol{A}^{-1},
\end{equation}
where $\boldsymbol{A}^{-1}$ is a diagonal matrix with diagonal elements $$\left(\left\{W_i\sigma_i^2 / M \right\}_{1\leq i\leq N}, \left\{\tilde{\sigma}^2/N\right\}_{1\leq j\leq M}\right).$$ 
Therefore, we can derive the inverse of $2\times 2$ matrix $\boldsymbol{C}+\boldsymbol{U}^\top\boldsymbol{A}^{-1}\boldsymbol{U}$ as 
\begin{align*}
& \left(\boldsymbol{C}+\boldsymbol{U}^\top\boldsymbol{A}^{-1}\boldsymbol{U}\right)^{-1} = \left(\begin{array}{cc}
\frac{\sum_{i=1}^N W_i\sigma_i^{-2}}{M} & 1\\
1 & \frac{M}{N}\tilde{\sigma}^2 
\end{array}\right)^{-1} =-\frac{\sum_{i=1}^N \sigma_i^{-2}}{\sum_{i=1}^N W_i\tau_i^{-2}}\left(\begin{array}{cc}\frac{M^2 \tilde{\sigma}^2}{N} 
 & -M\\
-M & \sum_{i=1}^N W_i\sigma_i^{-2}
\end{array}\right).
\end{align*}
Further, let $\boldsymbol{W}$ be the $N\times 1$ column vector with $i$th element $W_i$, then we have 
\begin{equation*}
\boldsymbol{A}^{-1}\boldsymbol{U} = \left(\begin{array}{cc}
\boldsymbol{W}/M & 0_{N\times 1}\\
0_{M\times 1} & \tilde{\sigma}^2 /N\  1_{M\times 1} 
\end{array}\right).
\end{equation*}
Consequently, $\boldsymbol{A}^{-1}\boldsymbol{U}\left(\boldsymbol{C}+\boldsymbol{U}^\top\boldsymbol{A}^{-1}\boldsymbol{U}\right)^{-1}\boldsymbol{U}^\top\boldsymbol{A}^{-1} $ is equal to
\begin{align*}
- \left(\sum_{i=1}^N \tau_i^{-2}W_i \right)^{-1} \left(\begin{array}{cc}
\boldsymbol{W} \boldsymbol{W}^\top & -\boldsymbol{W} 1_{1\times M}\\
-  1_{M\times 1} \boldsymbol{W}^\top  & [N^{-1}\tilde{\sigma}^2] \sum_{i=1}^N W_i\sigma_i^{-2}\  1_{M\times M}
\end{array}\right).
\end{align*}
Finally, we arrive at the closed-form expression for $\boldsymbol{\Omega}^{-1} (\boldsymbol{\sigma}^2)$:
\begin{align*}
\left(\boldsymbol{\Omega}^{-1} (\boldsymbol{\sigma}^2) \right)_{i,i}
&= \frac{1}{M\sigma_i^{-2}+\tau_i^{-2}}\left\{1 + \frac{M \sigma_i^{-2} W_i}{\sum_{u=1}^N {\tau_u^{-2} W_u} }\right\},\\
\left(\boldsymbol{\Omega}^{-1} (\boldsymbol{\sigma}^2) \right)_{i,j+N} &= \left(\boldsymbol{\Omega}^{-1} (\boldsymbol{\sigma}^2) \right)_{j+N,i} =  - W_i \left(\sum_{u=1}^N \tau_u^{-2} W_u\right)^{-1},\\
\left(\boldsymbol{\Omega}^{-1} (\boldsymbol{\sigma}^2) \right)_{j+N,j+N} 
 &= \left(\sum_{u=1}^N \sigma_u^{-2}\right)^{-1}\left\{1 + \frac{\sum_{i=1}^N W_i\ \sigma_i^{-2}}{\sum_{u=1}^N W_u\ \tau_u^{-2} }  \right\}.
\end{align*}

\section{Likelihood Method}
\label{section:frequentistmethod}

\subsection{MLEs and Their Asymptotic Variances}
\label{appendix:mlecompute}

Note that the variance-covariance matrix of the MLEs $\{\hat{\boldsymbol{B}}, \hat{\boldsymbol{G}}\}$ is in fact $\boldsymbol{\Omega}^{-1}(\boldsymbol{\sigma}^2)$ as defined in~\eqref{eqn:expressionomegasigma}. Therefore, we have the following proposition.

\begin{proposition}
\label{proposition:variance_BG_known_var}
 If all detectors measure all objects, i.e., $J_i = \{1,\ldots,M\}$, $I_j=\{1,\ldots,N\}$ and $\{\sigma_i^2, \tau_i^2\}$ are known constants, then the variances of $\{\hat{B}_i\}$, $\{\hat{G}_j\}$ are given by 
\begin{equation}
\label{eqn:var_BG_known_var}
\V(\hat{G}_j)  = \left[\sum_{i=1}^N \sigma_i^{-2}\right]^{-1} \ \mathcal{S}_G,\quad 
\V(\hat{B}_i) = \left[M\sigma_i^{-2} + {\tau}_i^{-2}\right]^{-1} \ \mathcal{S}_{B}^{(i)},
\end{equation}
where the inflation factors $\mathcal{S}_G, \{\mathcal{S}_B^{(i)}\}$ are given by
\begin{align*}
\mathcal{S}_{G}&= 1+\frac{\sum_{i=1}^{N} \sigma_i^{-2}W_i}{\sum_{i=1}^N \tau_i^{-2} W_i},\quad
\mathcal{S}_{B}^{(i)} = 1+ \frac{M \sigma_i^{-2} W_i}{\sum_{u=1}^N \tau_u^{-2} W_u}.
\end{align*}
Moreover, we have ${\rm Cov} (\hat{B}_i, \hat{G}_j) = - W_i \left[\sum_{k=1}^N \tau_k^{-2} W_k \right]^{-1}$.
\end{proposition}

\begin{remark}
Under the additive model, $B_i$ and $G_j$ are negatively correlated for all $i, j$. The asymptotic variances of $\hat{B}_i$ and $\hat{G}_j$ can be written as 
 \[\V(\hat{G}_j) = \V(\tilde{G}_j) \mathcal{S}_G,\quad \V(\hat{B}_i) = \V(\tilde{B}_i) \mathcal{S}_B^{(i)},\]
where $\V(\tilde{G}_j) = \left[\sum_{i=1}^N \sigma_i^{-2}\right]^{-1}$ is the inverse precision, i.e., asymptotic covariance, of $\hat{G}_j$ when the $B_i$ are known constants; $\V(\tilde{B}_i) = [M \sigma_i^{-2} + \tau_i^{-2}]^{-1}$ is the inverse precision, i.e., asymptotic covariance, of $\hat{B}_i$ when the $G_j$ are known constants. The inflation factors $\mathcal{S}_G$ and $\mathcal{S}_B^{(i)}$ adjust for the fact that none of the $B_i$ or the $G_j$ are known.
\end{remark}

\newtheorem{corollary}{Corollary}
Proposition~\ref{proposition:variance_BG_known_var} directly yields the following asymptotic results as $N, M\rightarrow\infty$.
\begin{corollary}
\label{coro:asymptoticvarianceMLE}
If $\{\sigma_i/\tau_i\}$ are uniformly bounded from below and above by finite positive constants, and $\sum_{i=1}^N \sigma_i^{-2} / N$ converges to a positive constant as $N\rightarrow\infty$, then for all $i, j$, as $N, M\rightarrow\infty$,
\begin{equation*}
\V(\hat{G}_j) =O(N^{-1}),\  \V(\hat{B}_i) = O(N^{-1}+M^{-1}),\  {\rm Cov}(\hat{B}_i,\hat{G}_j) = - O(N^{-1}). 
\end{equation*}
Specifically, when $\tau=\tau_1=\cdots=\tau_N$ and $\sigma=\sigma_1=\cdots=\sigma_N$, \eqref{eqn:var_BG_known_var} simplifies to
\begin{equation*}
\V(\hat{G}_j) = \frac{\sigma^2}{N},\  \V(\hat{B}_i) = \frac{1}{M\sigma^{-2}+\tau^{-2}} \left(1+\frac{M\sigma^{-2}}{N\tau^{-2}}\right), \ {\rm Cov}(\hat{B}_i, \hat{G}_j) = -\frac{\tau^2}{N}.
\end{equation*}
\end{corollary}

\begin{remark}
The results above show that the asymptotic variances for $\{B_i\}$ and $\{G_j\}$ are not `exchangeable' (i.e., switching $\boldsymbol{B}$ and $\boldsymbol{G}$ and correspondingly $N$ and $M$), mainly for three reasons: first, for each $B_i$ we assign an informative prior $\mathcal{N}(b_i,\tau_i^2)$ whereas for each $G_j$ we assign a flat prior on the real line; second, for each instrument $i$, besides $B_i$, we also need to estimate $\sigma_{i}^2$; last, the measurement uncertainty depends only on the instrument but not on the sources (recall that $\sigma_{ij}^2=\sigma_i^2$ for all $i, j$). 
\end{remark}

\subsection{Goodness-of-fit}
\label{subsec:goodnessoffitregression}

We now give a goodness-of-fit test statistic for the random-effect regression model. Under the model (\ref{eqn:lognormalmodel}), we have the following normalized residual sum of squares:
\begin{equation}
\label{eqn:defnormalizedsse}
T(\boldsymbol{B}, \boldsymbol{G}):= \sum_{i=1}^N \frac{(b_i - B_i)^2}{\tau_i^2} + \sum_{i=1}^N\sum_{j=1}^M \frac{\left(y_{ij}' - B_i - G_j\right)^2}{\sigma_{i}^2}.
\end{equation}
We see this sum of squares has two parts. The first part involves $\{b_i\}$ only, measuring how good the prior means are relative to the prior variances $\{\tau_i\}$.
The second part depends on $\{y_{ij}\}$ only, and it will allow us to access how good the fitted $\boldsymbol{B, G}$ are relative to the sampling variances $\boldsymbol{\sigma}^2$. Here we put them together as an overall model check, but one can certainly use them individually if one wants to check the prior distribution and likelihood model separately.

\begin{theorem}
\label{theorem:RSS}
When the variances $\sigma_{i}^2,\tau_i^2$ are known and we insert the MLEs of $B_i$ and $G_j$ into (\ref{eqn:defnormalizedsse}), we obtain
$T(\hat{\boldsymbol{B}}, \hat{\boldsymbol{G}}) \sim \chi^2_{NM-M}.$
\end{theorem}
\begin{proof} This conclusion regarding $\chi^2$ distribution follows from standard results on residual sum of squares of linear regression with Gaussian error. To figure out the correct degrees of freedom, we have ($NM+N$) independent observations in total, but  with  $N+M$ parameters. Therefore, the degrees of freedom for the residual sum of squares is $NM-M$. 
\end{proof}

With unknown variances we do not have a closed-form distribution of $T$ as defined in formula~(\ref{eqn:defnormalizedsse}). Heuristically, we invoke  the standard large-sample arguments and to continuously use the $\chi^2$ approximation, but reduce the degrees of freedom to $MN-M-N$ to count for the number of estimated variance parameter $\{\sigma_{i}^2\}$. The resulting p-values of the fitted data in Sections~\ref{section:2xmmhmsdata} and~\ref{section:hmsdata} are not significant.

\section{More Simulation Results Under Misspecified Models}
\label{appendix:robustness_simulation}

In Simulations~\upperRomannumeral{4} and~\upperRomannumeral{5}, we generate data as $c_{ij} =\lambda_{ij} X_{ij}$, where  $X_{ij}\sim \text{Poisson}(A_i F_j)$, and independently $\lambda_{ij}\sim \text{Uniform}[0.8, 1.2]$ for Simulation~\upperRomannumeral{4} and $\lambda_{ij}\sim \text{Uniform}[0.4, 1.6]$ for Simulation~\upperRomannumeral{5}. In Simulations~\upperRomannumeral{6} and~\upperRomannumeral{7}, we generate data from 
$c_{ij} \sim \text{Poisson}(\lambda_{ij} A_i F_j),$
where the $\lambda_{ij}$ are randomly generated from the uniform distribution on $[0.8, 1.2]$. The other parameters are set to be the same as in Simulation~\upperRomannumeral{2} except that $\beta=0.01$ for these simulations. Simulations~\upperRomannumeral{6} and~\upperRomannumeral{7} resemble the cases where the true model is Poisson and the estimation of 
 $T_{ij}$ is volatile, whereas Simulations~\upperRomannumeral{4} and~\upperRomannumeral{5} resemble the cases that happen in practice, where the photon counts are multiplied by an adjustment factor, such as $\hat T_{ij}^{-1}$, as with the data pre-processing step for the XCAL data.

\begin{figure}[tbph]
\centering
\includegraphics[width=0.9\textwidth]{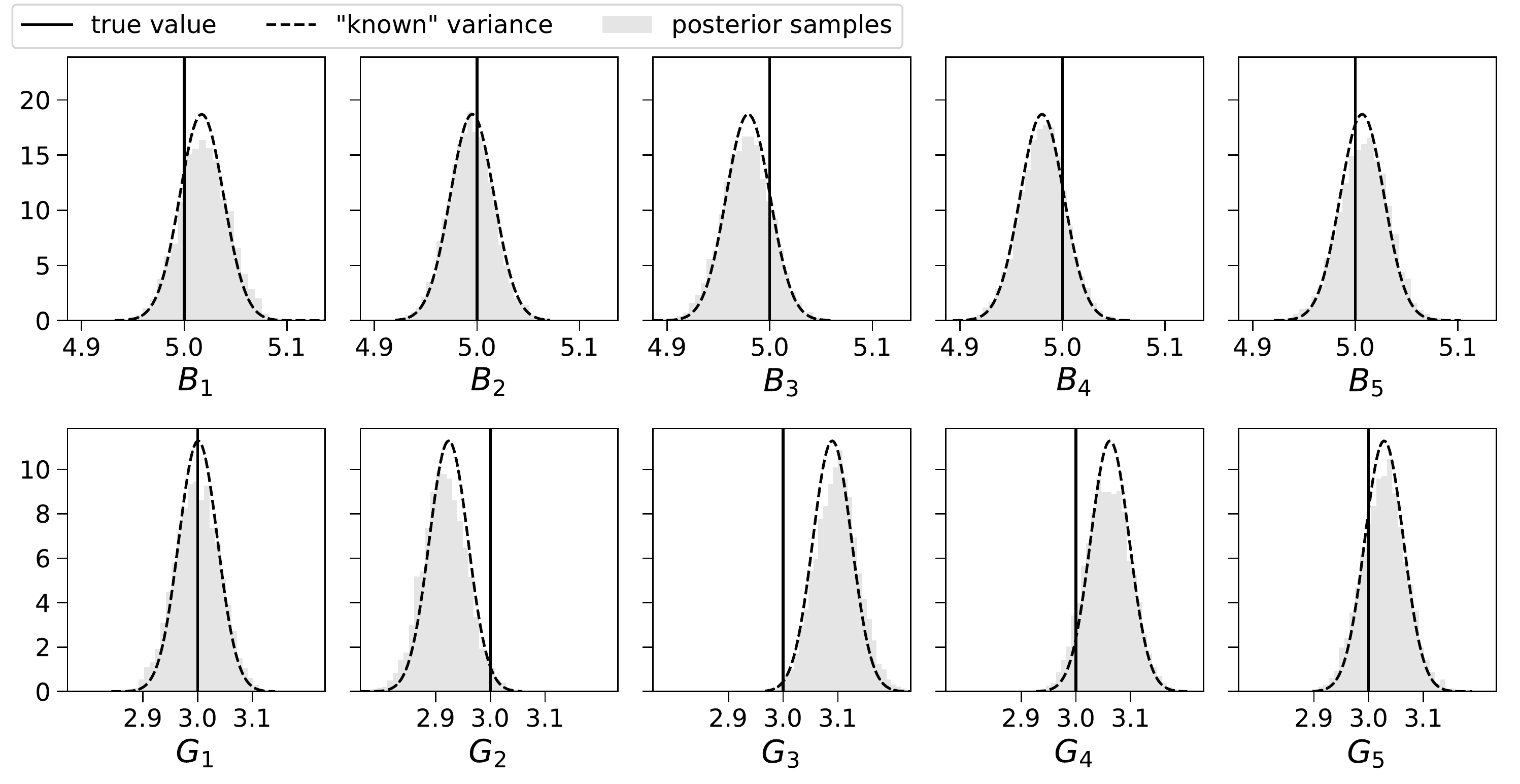}\\
\includegraphics[width=0.9\textwidth]{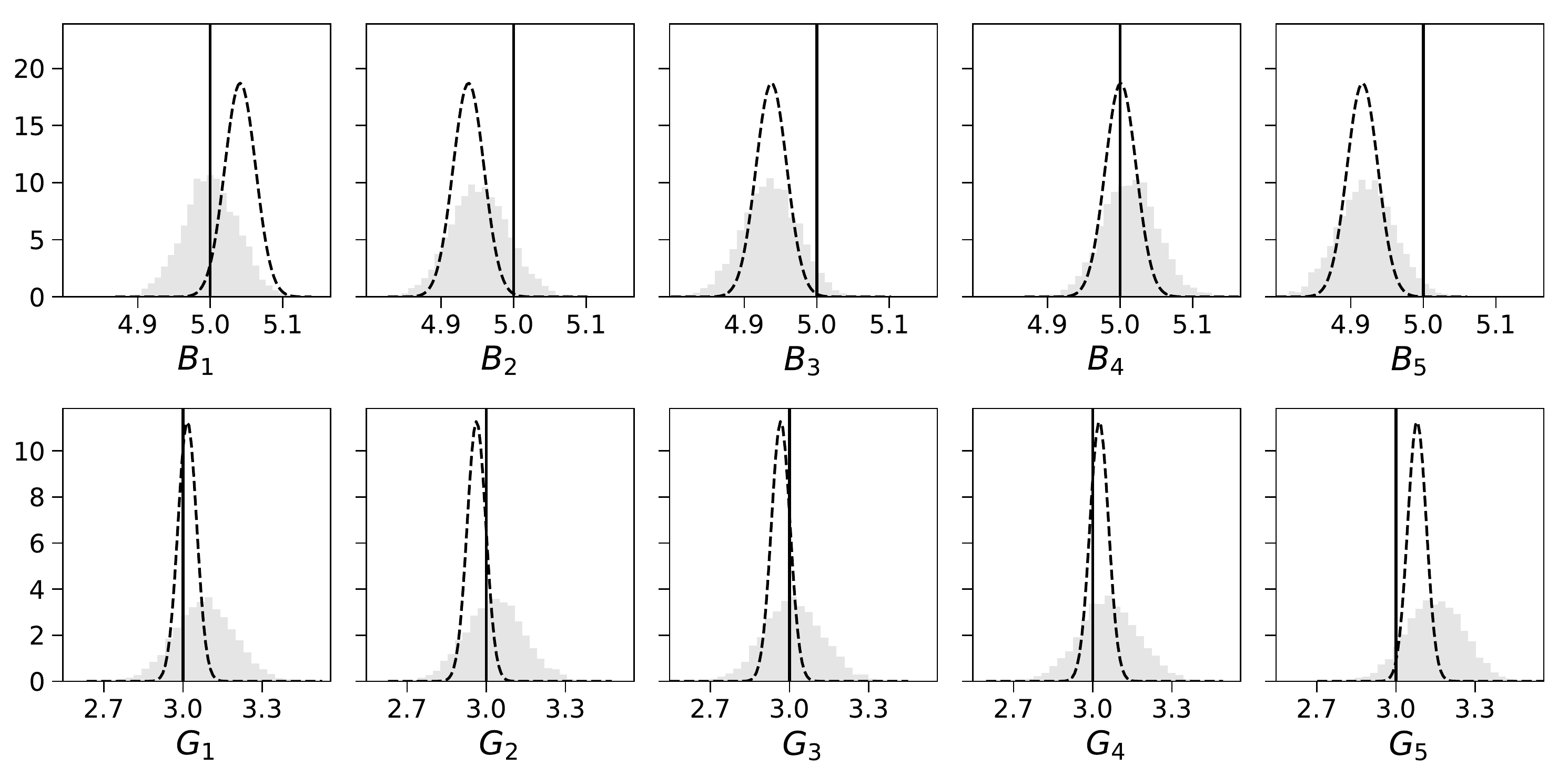}\caption{Simulations~\upperRomannumeral{4} (rows 1 \& 2) and~\upperRomannumeral{5} (rows 3 \& 4). The gray shades are the posterior distributions of $\{B_i\}_{i=1}^5$ (rows 1 \& 3) and $\{G_j\}_{j=1}^5$ (rows 2 \& 4) fitted with unknown variances. The solid vertical black lines denote the true values. The black dashed density curves on top of the histograms denote the true posterior densities of  
$\{B_i\}$ and the $\{G_j\}$ with `known' variances $\sigma_i^2=0.1^2$.}
\label{fig:simmisspecifiedpoissonconst}
\end{figure}

Figure~\ref{fig:simmisspecifiedpoissonconst} gives the results of Simulations~\upperRomannumeral{4} and~\upperRomannumeral{5}. Figure~\ref{fig:simmisspecifiedpoisson} gives the results of Simulation~\upperRomannumeral{6} with smaller counts ($B_i=1$ and $G_j=3$) and~\upperRomannumeral{7} with larger counts ($B_i=5$ and $G_j=3$) under this scenario. It shows with large Poisson counts, controlling the uncertainty in the multiplicative constant can possibly lead to reasonably good results. Thus, even with compounded model misspecification, the log-Normal model is able to provide reasonable, though not as precise, results, as compared with the correctly-specified case. However, when the misspecified ``known constant'' is highly variable, the fit result is not as satisfactory; plugging in a ``guesstimated'' $\sigma_i$ in this case can give disastrously optimistic but biased results.

\begin{figure}[tbph]
\centering
\includegraphics[width = 0.85\textwidth]{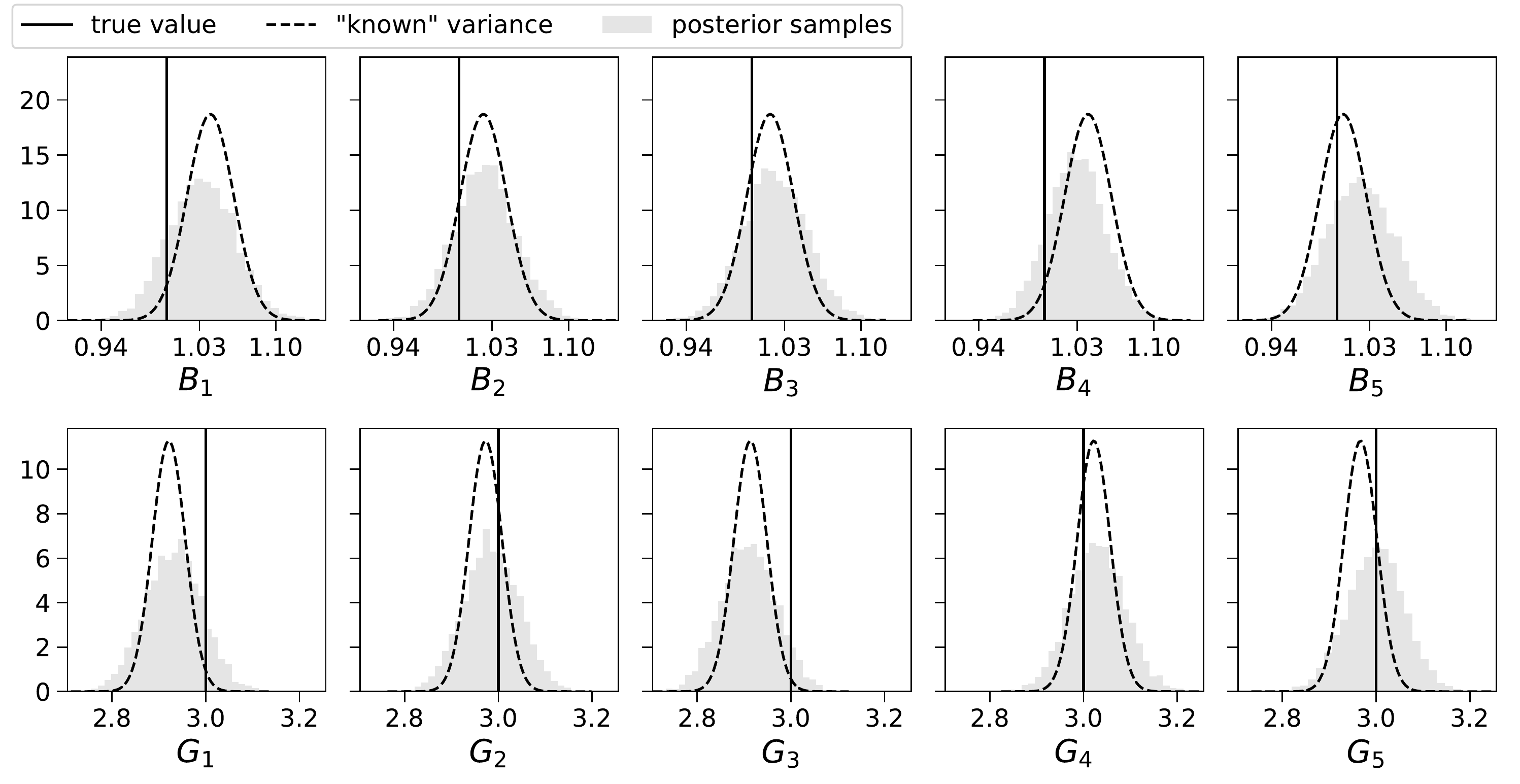}\\
\includegraphics[width = 0.85\textwidth]{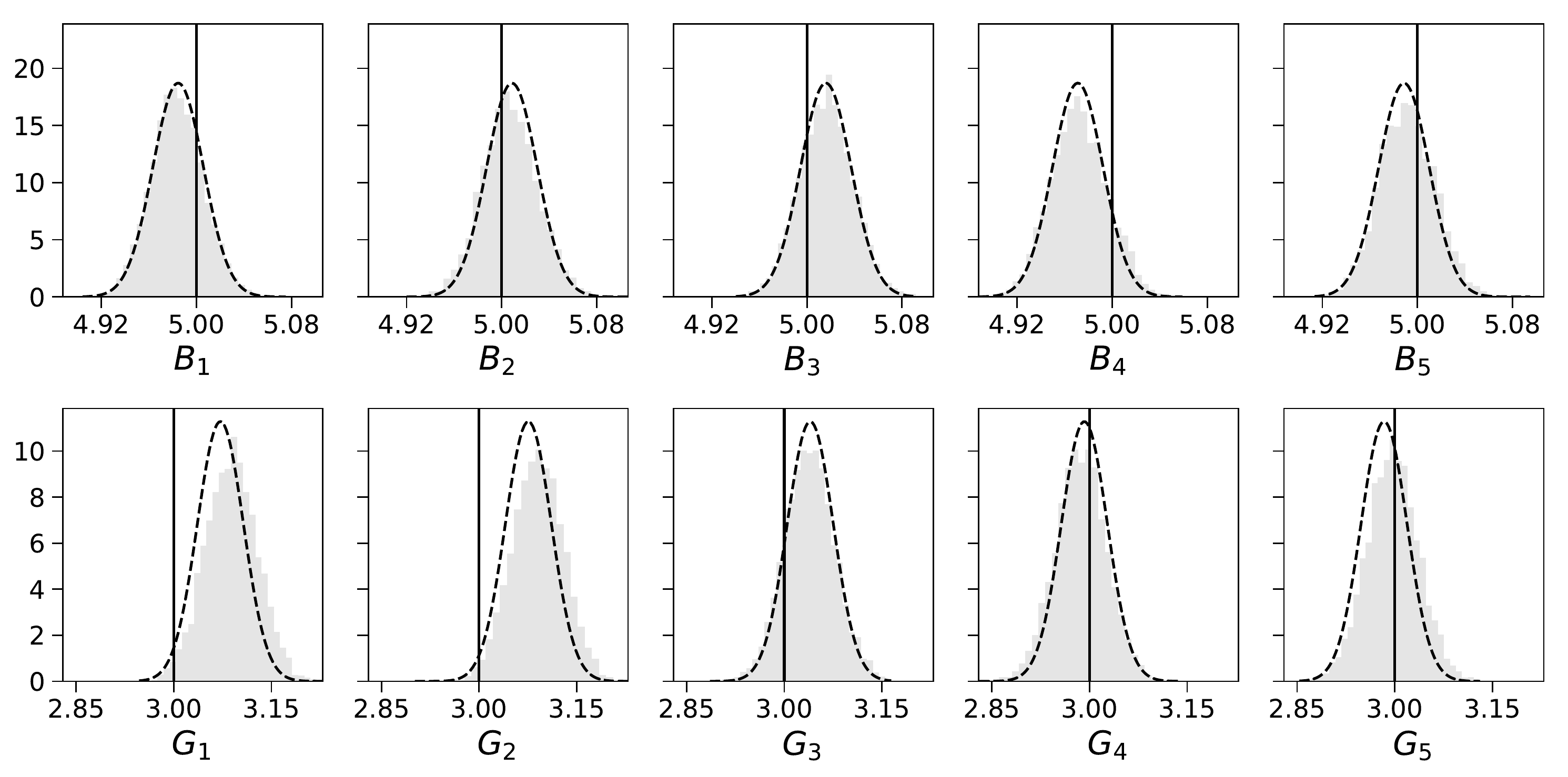}
\caption{Same as Figure~\ref{fig:simmisspecifiedpoissonconst} but with Simulations~\upperRomannumeral{6} (rows 1 \& 2) and~\upperRomannumeral{7} (rows 3 \& 4).}
\label{fig:simmisspecifiedpoisson}
\end{figure}

\end{document}